\documentclass[acmsmall,10pt,table,xcdraw]{acmart}
\acmConference[OOPSLA'18]{ACM SIGPLAN Conference on Object Oriented Programming, Systems, Languages, and Applications}{November 04--09, 2018}{Boston, MA, USA}
\acmYear{2018}
\acmISBN{} % \acmISBN{978-x-xxxx-xxxx-x/YY/MM}
\acmDOI{} % \acmDOI{10.1145/nnnnnnn.nnnnnnn}
\startPage{1}

\setcopyright{none}
\usepackage{booktabs}   %% For formal tables:
\usepackage{subcaption} %% For complex figures with subfigures/subcaptions
\usepackage{xcolor}
\usepackage{tikz}
\usepackage[inline]{enumitem}
\usepackage[ruled,vlined,resetcount,linesnumbered]{algorithm2e}
\usepackage{multicol}
\usepackage{amsthm}
\usepackage{fancyvrb}

\include{macros}

\makeatletter
\newtheorem*{rep@theorem}{\rep@title}
\newcommand{\newreptheorem}[2]{%
\newenvironment{rep#1}[1]{%
 \def\rep@title{#2 \ref{##1}}%
 \begin{rep@theorem}}%
 {\end{rep@theorem}}}
\makeatother

\newreptheorem{theorem}{Theorem}
\newreptheorem{lemma}{Lemma}

\theoremstyle{remark}
\newtheorem*{remark}{Remark}

\begin{document}

%% Title information
\title[What Happens-After the First Race?]{What Happens-After the First Race?}         %% [Short Title] is optional;
                                        %% when present, will be used in
                                        %% header instead of Full Title.
% \titlenote{with title note}             %% \titlenote is optional;
                                        %% can be repeated if necessary;
                                        %% contents suppressed with 'anonymous'
\subtitle{Enhancing the Predictive Power of Happens-Before Based Dynamic Race Detection}                     %% \subtitle is optional
% \subtitlenote{with subtitle note}       %% \subtitlenote is optional;
                                        %% can be repeated if necessary;
                                        %% contents suppressed with 'anonymous'

%% Author information
%% Contents and number of authors suppressed with 'anonymous'.
%% Each author should be introduced by \author, followed by
%% \authornote (optional), \orcid (optional), \affiliation, and
%% \email.
%% An author may have multiple affiliations and/or emails; repeat the
%% appropriate command.
%% Many elements are not rendered, but should be provided for metadata
%% extraction tools.

%% Author with single affiliation.
\author{Umang Mathur}
% \authornote{with author1 note}          %% \authornote is optional;
                                        %% can be repeated if necessary
\orcid{0000-0002-7610-0660}             %% \orcid is optional
\affiliation{
%   \position{Position2a}
  \department{Department of Computer Science}              %% \department is recommended
  \institution{University of Illinois, Urbana Champaign}            %% \institution is required
  % \streetaddress{Street1 Address1}
  % \city{City1}
  % \state{State1}
  % \postcode{Post-Code1}
  \country{USA}                    %% \country is recommended
}
\email{umathur3@illinois.edu}          %% \email is recommended

%% Author with two affiliations and emails.
\author{Dileep Kini}
\affiliation{
%   \position{Position2a}
  % \department{Department of Computer Science}              %% \department is recommended
  \institution{Akuna Capital LLC}            %% \institution is required
  % \streetaddress{Street1 Address1}
  % \city{City1}
  % \state{State1}
  % \postcode{Post-Code1}
  \country{USA}                    %% \country is recommended
}
\email{dileeprkini@gmail.com}          %% \email is recommended

\author{Mahesh Viswanathan}
\affiliation{
%   \position{Position2a}
  \department{Department of Computer Science}              %% \department is recommended
  \institution{University of Illinois, Urbana Champaign}            %% \institution is required
  % \streetaddress{Street1 Address1}
  % \city{City1}
  % \state{State1}
  % \postcode{Post-Code1}
  \country{USA}                    %% \country is recommended
}
\email{vmahesh@illinois.edu} 
% \authornote{with author2 note}          %% \authornote is optional;
                                        %% can be repeated if necessary
% \orcid{nnnn-nnnn-nnnn-nnnn}             %% \orcid is optional
% \affiliation{
%   \position{Position2a}
%   \department{Department2a}             %% \department is recommended
%   \institution{Institution2a}           %% \institution is required
%   \streetaddress{Street2a Address2a}
%   \city{City2a}
%   \state{State2a}
%   \postcode{Post-Code2a}
%   \country{Country2a}                   %% \country is recommended
% }
% \email{first2.last2@inst2a.com}         %% \email is recommended
% \affiliation{
%   \position{Position2b}
%   \department{Department2b}             %% \department is recommended
%   \institution{Institution2b}           %% \institution is required
%   \streetaddress{Street3b Address2b}
%   \city{City2b}
%   \state{State2b}
%   \postcode{Post-Code2b}
%   \country{Country2b}                   %% \country is recommended
% }
% \email{first2.last2@inst2b.org}         %% \email is recommended

%% Abstract
%% Note: \begin{abstract}...\end{abstract} environment must come
%% before \maketitle command
\begin{abstract}
%!TEX root = main.tex

Dynamic race detection is the problem of determining if an observed
program execution reveals the presence of a data race in a
program. The classical approach to solving this problem is to detect
if there is a pair of conflicting memory accesses that are unordered
by Lamport's happens-before (HB) relation. 
HB based race detection is known to not report false positives, i.e.,
it is sound.
However, the soundness guarantee of HB only promises 
that the first pair of unordered, conflicting events is a 
\emph{schedulable} data race. That is, there
can be pairs of HB-unordered conflicting data accesses that are
not schedulable races because there is no reordering of
the events of the execution, where the events in race can be executed
immediately after each other. We introduce a new partial order, called
schedulable happens-before (SHB) that exactly characterizes the pairs of
schedulable data races --- every pair of conflicting data accesses
that are identified by SHB can be scheduled, and every HB-race that can be
scheduled is identified by SHB. 
Thus, the SHB partial order is truly sound.
We present
a linear time, vector clock algorithm to detect schedulable races
using SHB. 
Our experiments demonstrate the value of our algorithm for dynamic race
detection --- SHB incurs only little performance overhead and can scale to
executions from real-world software applications without
compromising soundness.
\end{abstract}

%% 2012 ACM Computing Classification System (CSS) concepts
%% Generate at 'http://dl.acm.org/ccs/ccs.cfm'.

%% End of generated code

%% Keywords
%% comma separated list
  %% \keywords are mandatory in final camera-ready submission

%% \maketitle
%% Note: \maketitle command must come after title commands, author
%% commands, abstract environment, Computing Classification System
%% environment and commands, and keywords command.
\maketitle

\section{Introduction}
\label{sec:intro}
%!TEX root = main.tex

The presence of data races in concurrent software is the most common
indication of a programming error. Data races in programs can result
in nondeterministic behavior that can have unintended
consequences. Further, manual debugging of such errors 
is prohibitively difficult owing to nondeterminism.
Therefore, automated detection and elimination of data races is an
important problem that has received widespread attention from the
research community. 
% Race detection techniques can be broadly
% classified as either being
% static~\cite{Naik:2006:ESR:1133255.1134018,pratikakis11locksmith,Radoi:2013:PSR:2483760.2483765,racerx,voung2007relay,heisenbugs,Yahav:2001:VSP:373243.360206,echo},
% i.e., those that analyze source code to detect races, or
% dynamic~\cite{savage1997eraser,Pozniansky:2003:EOD:966049.781529,cp2012,wcp2017,rvpredict,elmas2007goldilocks,Said2011,fasttrack,vonPraun:2001:ORD:504311.504288,Sen:2008:RDR:1375581.1375584,Huang2016,ipa2016,SPA2009}, i.e., those that examine a single execution of the
% program to discover a data race in the program.
Dynamic race detection techniques examine a single execution of a
concurrent program to discover a data race in the program.
In this paper we focus on dynamic race detection.
% \ucomment{Put more blabber}

Dynamic race detection may either be sound or unsound. Unsound
techniques, like lockset based methods~\cite{savage1997eraser}, have
low overhead but they report potential races that are spurious. Sound
techniques~\cite{lamport1978time,Mattern1988,Said2011,rv2014,cp2012,wcp2017},
on the other hand, never report the presence of a data race, if none
exist. The most popular, sound technique is based on computing the
\emph{happens-before} (HB) partial order~\cite{lamport1978time} on the
events of the trace, and declares a data race when there is a pair of
conflicting events (reads/writes to a common memory location performed
by different threads, at least one of which is a write operation) that
are unordered by the partial order. There are two reasons for the
popularity of the HB technique. First, because it is sound, it does
not report false positives. Low false positive rates are critical for
the wide-spread use of debugging techniques~\cite{threadsanitizer,developersRace14}. Second, even
though HB-based algorithms may miss races detected by other sound
techniques~\cite{Said2011,rv2014,cp2012,wcp2017}, they have the lowest
overhead among sound techniques. Many
improvements~\cite{Pozniansky:2003:EOD:966049.781529,fasttrack,elmas2007goldilocks}
to the original vector clock algorithm~\cite{Mattern1988} have helped
reduce the overhead even further.

%!TEX root = main.tex

\begin{figure}[t]
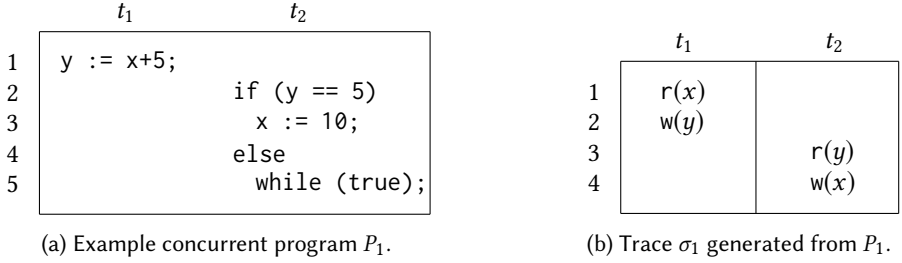

\centering
\begin{subfigure}{.5\textwidth}
  \centering
  \figprogram{2}{
    \figstmt{1}{0}{y := x+5;}
    \figstmt{2}{0}{if (y == 5)}
    \figstmt{2}{1}{x := 10;}
    \figstmt{2}{0}{else}
    \figstmt{2}{1}{while (true);}
  }
  \caption{Example concurrent program $P_1$.}
  \label{fig:program1}
\end{subfigure}%
\begin{subfigure}{.5\textwidth}
  \centering
  \vspace{0.15in}
  \execution{2}{
    \figev{1}{\rd(x)}
    \figev{1}{\wt(y)}
    \figev{2}{\rd(y)}
    \figev{2}{\wt(x)}
  }
  \caption{Trace $\tr_1$ generated from $P_1$.}
  \label{fig:trace1}
\end{subfigure}
\caption{Concurrent program $P_1$ and its sample execution
  $\tr_1$. Initially $x = y = 0$.}
\label{fig:example1}
\end{figure}

However, HB-based dynamic analysis tools suffer from some
drawbacks. Recall that a program has a data race, if there is some
execution of the program where a pair of conflicting data accesses are performed
consecutively. Even though HB is a sound technique, its soundness
guarantee is only limited to the \emph{first} pair of unordered
conflicting events; a formal definition of ``first'' unordered pair is
given later in the paper. Thus, a trace may have many HB-unordered
pairs of conflicting events (popularly called \emph{HB-races}) that do
not correspond to data races. To see this, consider the example
program and trace shown in Fig.~\ref{fig:example1}. The trace
corresponds to first executing the statement of thread $t_1$, before
executing the statements of thread $t_2$.  The statement $\texttt{y :=
x + 5}$ requires first reading the value of $x$ (which is $0$) and
then writing to $y$. Recall that HB orders (i) two events performed by
the same thread, and (ii) synchronization events performed by
different threads, in the order in which they appear in the
trace. Using $e_i$ to denote the $i$th event of the trace, in this
trace since there are no synchronization events, both $(e_1,e_4)$ and
$(e_2,e_3)$ are in HB race. Observe that while $e_2$ and $e_3$ can
appear consecutively in a trace (as in Fig.~\ref{fig:trace1}), there
is no trace of the program where $e_1$ and $e_4$ appear
consecutively. Thus, even though the events $e_1$ and $e_4$ are
unordered by HB,
they do not constitute a data race.

%!TEX root = main.tex

\begin{figure}[h]
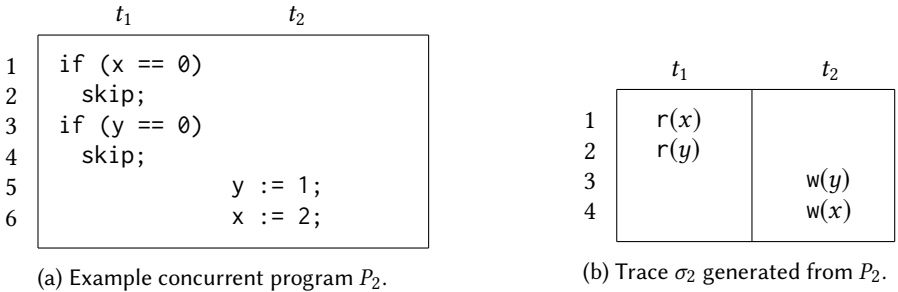

\centering
\begin{subfigure}{.5\textwidth}
  \centering
  \figprogram{2}{
    \figstmt{1}{0}{if (x == 0)}
    \figstmt{1}{1}{skip;}
    % \figstmt{1}{0}{else}
    % \figstmt{1}{1}{while(true);}
    \figstmt{1}{0}{if (y == 0)}
    \figstmt{1}{1}{skip;}
     % \figstmt{1}{0}{else}
    % \figstmt{1}{1}{while(true);}
    \figstmt{2}{0}{y := 1;}
    \figstmt{2}{0}{x := 2;}
  }
  \caption{Example concurrent program $P_2$.}
  \label{fig:program2}
\end{subfigure}%
\begin{subfigure}{.5\textwidth}
  \centering
  \vspace{0.25in}
  \execution{2}{
    \figev{1}{\rd(x)}
    \figev{1}{\rd(y)}
    \figev{2}{\wt(y)}
    \figev{2}{\wt(x)}
  }
  \caption{Trace $\tr_2$ generated from $P_2$.}
  \label{fig:trace2}
\end{subfigure}
\caption{Concurrent program $P_2$ and its sample execution
  $\tr_2$. Initially $x = y = 0$.}
\label{fig:example2}
\end{figure}

As a consequence, developers typically fix the first race discovered,
re-run the program and the dynamic race detection algorithm, and
repeat the process until no races are discovered. This approach to bug
fixing suffers from many disadvantages. First, running race detection
algorithms can be expensive~\cite{developersRace14}, and so running
them many times is a significant overhead. Second, even though only
the first HB race is guaranteed to be a real race, it doesn't mean
that it is the \emph{only} HB race that is real. Consider the example
shown in Fig.~\ref{fig:example2}. In the trace $\tr_2$ (shown in
Fig.~\ref{fig:trace2}), both pairs $(e_1,e_4)$ and $(e_2,e_3)$ are in
HB-race. $\tr_2$ demonstrates that $(e_2,e_3)$ is a valid data race
(because they are scheduled consecutively). But $(e_1,e_4)$ is also a
valid data race. This can be seen by first executing $\texttt{y :=
1;}$ in thread $t_2$, followed by $\texttt{if (x == 0) skip;}$ in
thread $t_1$, and then finally $\texttt{x := 2;}$ in $t_2$. The
approach of fixing the first race, and then re-executing and
performing race detection, not only unnecessarily ignores the race
$(e_1,e_4)$, but it might miss it completely because $(e_1,e_4)$ might
not show up as a HB race in the next execution due to the inherent
nondeterminism when executing multi-threaded programs.  As a result,
most practical race detection tools including
ThreadSanitizer~\cite{threadsanitizer}, Helgrind~\cite{helgrind}
and \fasttrack~\cite{fasttrack} report more than one race, even if
those races are likely to be false, to give software developers the
opportunity to fix more than just the first race.
In Appendix~\ref{app:false_races}, we illustrate
this observation on four practical dynamic race detection tools
based on the happens-before partial order. Each of these tools
resort to  na\"ively reporting races beyond the first race and produce
false positives as a result.

The central question we would like to explore in this paper is, can we
detect multiple races in a given trace, soundly? One approach would be
to mimic the software developer's strategy in using HB-race detectors
--- every time a race is discovered, force an order between the two
events constituting the race and then analyze the subsequent
events. This ensures that the HB soundness theorem then applies to
the \emph{next} race discovered, and so on. Such an algorithm can be
proved to only discover valid data races. For example, in trace
$\tr_1$ (Fig.~\ref{fig:example1}), after discovering the race
$(e_2,e_3)$ assume that the events $e_2$ and $e_3$ are ordered when
analyzing events after $e_3$ in the trace. By this algorithm, when we
process event $e_4$, we will conclude that $(e_1,e_4)$ are not in race
because $e_1$ comes before $e_2$, $e_2$ has been force ordered before
$e_3$, and $e_3$ is before $e_4$, and so $e_1$ is ordered before
$e_4$. However, force ordering will miss valid data races present in
the trace. Consider the trace $\tr_2$ from
Fig.~\ref{fig:example2}. Here the force ordering algorithm will only
discover the race $(e_2,e_3)$ and will miss $(e_1,e_4)$ which is a
valid data race. Another approach~\cite{rv2014}, is to search for a
reordering of the events in the trace that respects the data
dependencies amongst the read and write events, and the effect of
synchronization events like lock acquires and releases. Here one
encodes the event dependencies as logical constraints, where the
correct reordering of events corresponds to a satisfying truth
assignment. The downside of this approach is that the SAT formula
encoding event dependencies can be huge even for a trace with a few
thousand events. Typically, to avoid the prohibitive cost of
determining the satisfiability of such a large formula, the trace is
broken up into small ``windows'', and the formula only encodes the
dependencies of events within a window. In addition, solver timeouts
are added to give up the search for another reordering. As a
consequence this approach can miss many data races in practice (see
our experimental evaluation in Section~\ref{sec:exp}).

In this paper, we present a new partial order on events in an
execution that we call \emph{schedulable happens-before} (SHB) to
address these challenges. Unlike recent attempts~\cite{cp2012,wcp2017}
to \emph{weaken} HB to discover more races, SHB is
a \emph{strengthening} of HB --- some HB unordered events, will be
ordered by SHB. However, the first HB race (which is guaranteed to be
a real data race by the soundness theorem for HB) will also be SHB
unordered. Further, every race detected using SHB is a valid,
schedulable race. In addition, we prove that, not only does SHB
discover every race found by the na\"{i}ve force ordering algorithm
and more (for example, SHB will discover both races in
Fig.~\ref{fig:example2}), it will detect \emph{all HB-schedulable}
races. The fact that SHB detects precisely the set of HB-schedulable
races, we hope, will make it popular among software developers because
of its enhanced predictive power per trace and the absence of false
positives.

We then present a simple vector clock based algorithm for detecting
all SHB races. Because the algorithm is very close to the usual HB
vector clock algorithm, it has a low overhead. We also show how to
adapt existing improvements to the HB algorithm, like the use
of \emph{epochs}~\cite{fasttrack}, into the SHB algorithm to lower
overhead. We believe that existing HB-based detectors can be easily
modified to leverage the greater power of SHB-based analysis. 
%Like
%HB, our SHB-based algorithm uses $O(\log n)$ space for traces that
%have a constant number of variables, locks, and threads. While it
%detects all HB-schedulable write-write and write-read races, it does
%not detect all HB-schedulable read-write races. It turns out that
%detecting every HB-schedulable read-write race is a computationally
%harder problem; we show that any algorithm detecting every read-write
%race must use linear space, i.e., effectively store the entire trace.
%
We have implemented our SHB algorithm and analyzed its performance on
standard benchmarks. Our experiments demonstrate that (a) many HB
unordered conflicting events may not be valid data races, (b) there
are many valid races missed by the na\"{i}ve force ordering algorithm,
(c) SHB based analysis poses only a little overhead as compared to HB
based vector clock algorithm, and (d) improvements like the use of
epochs, are effective in enhancing the performance of SHB analysis.

% \subsection{Related Work}
% \ucomment{Put detailed related work.}

% \subsection{Outline}
The rest of the paper is organized as follows:
% \begin{itemize}
	% \item 
	Section~\ref{sec:prelim} introduces notations and definitions relevant for the paper.
	% \item 
	In Section~\ref{sec:shb}, we introduce the partial order SHB and present an exact 
	characterization of schedulable races using this partial order.
	% \item 
	In Section~\ref{sec:algo}, we describe a vector clock algorithm for detecting
	schedulable races based on SHB. We then show how to incorporate epoch-based
	optimizations to this vector clock algorithm.
	% \item 
	Section~\ref{sec:exp} describes our experimental evaluation.
	% \item 
	We discuss relevant related work in Section~\ref{sec:related} and present
	concluding remarks in Section~\ref{sec:conclusion}.
% \end{itemize}

\section{Preliminaries}
\label{sec:prelim}
%!TEX root = main.tex
In this section, we will fix notation and present some definitions
that will be used in this paper.

\vspace*{0.1in}
\noindent
{\bf Traces.}  We consider concurrent programs under the sequential
consistency model.  Here, an execution, or trace, of a program is
viewed as an interleaving of operations performed by different
threads.  We will use $\tr$, $\tr'$ and $\tr''$ to denote traces.  For
a trace $\tr$, we will use $\threads{\tr}$ to denote the set of
threads in $\tr$. A trace is a sequence of events of the form $e
= \ev{t, op}$, where $t \in \threads{\tr}$, and $op$ can be one of
$\rd(x)$, $\wt(x)$ (read or write to memory location $x$),
$\acq(\lk)$, $\rel(\lk)$ (acquire or release of lock $\lk$) and
$\fork(u)$, $\join(u)$ (fork or join to some thread
$u$)~\footnote{Formally, each event in a trace is assumed to have a
unique event id. Thus, two occurences of a thread performing the same
operation will be considered \emph{different} events. Even though we
will implicitly assume the uniqueness of each event in a trace, to
reduce notational overhead, we do not formally introduce event ids.}.
To keep the presentation simple, we assume that locks are not
reentrant. However, all the results can be extended to the case when
locks are assumed to be reentrant. The set of events in trace $\tr$
will be denoted by $\events{\tr}$.  We will also use $\reads{\tr}(x)$
(resp. $\writes{\tr}(x)$) to denote the set of events that read
(resp. write) to memory location $x$.  Further $\reads{\tr}$
(resp. $\writes{\tr}$) denotes the union of the above sets over all
memory locations $x$. For an event $e \in \reads{\tr}(x)$,
the \emph{last write before} $e$ is the (unique) event
$e' \in \writes{\tr}(x)$ such that $e'$ appears before $e$ in the
trace $\tr$, and there is no event $e'' \in \writes{\tr}(x)$ between
$e'$ and $e$ in $\tr$. The last write before event
$e \in \reads{\tr}(x)$ maybe undefined, if there is no $\wt(x)$-event
before $e$. We denote the last write before $e$ by $\lw{\tr}(e)$. An
event $e = \ev{t_1, op}$ is said to be \emph{an event of thread $t$}
if either $t = t_1$ or
$op \in \{\fork(t), \join(t)\}$. The \emph{projection} of a trace
$\tr$ to a thread $t \in \threads{\tr}$ is the maximal subsequence of
$\tr$ that contains only events of thread $t$, and will be denoted by
$\proj{\tr}{t}$; thus an event $e = \ev{t, \fork(t')}$ (or $e
= \ev{t, \join(t')}$) belongs to both $\proj{\tr}{t}$ and
$\proj{\tr}{t'}$.  For an event $e$ of thread $t$, we denote by
$\ltho{\tr}(e)$ to be the last event $e'$ before $e$ in $\tr$ such
that $e$ and $e'$ are events of the same thread. Again,
$\ltho{\tr}(e)$ may be undefined for an event $e$.  The projection of
$\tr$ to a lock $\lk$, denoted by $\proj{\tr}{\lk}$, is the maximal
subsequence of $\tr$ that contains only acquire and release events of
lock $\lk$.  Traces are assumed to be well formed --- for every lock
$\lk$, $\proj{\tr}{\lk}$ is a prefix of some string belonging to the
regular language
$(\cup_{t \in \threads{\tr}} \ev{t, \acq(\lk)} \cdot \ev{t, \rel(\lk)}
)^*$.

\begin{figure}
\centering
\execution{4}{
  \figev{1}{\acq(\lk)}
  \figev{1}{\wt(x)}
  \figev{1}{\rel(\lk)}
  \figev{2}{\acq(\lk)}
  \figev{2}{\wt(x)}
  \figev{2}{\rel(\lk)}
  \figev{3}{\rd(x)}
  \figev{3}{\fork(t_4)}
  \figev{4}{\wt(x)}
  \figev{4}{\wt(x)}
  \figev{3}{\join(t_4)}
  \figev{3}{\rd(x)}
}
\caption{Trace $\tr_3$.}
\label{fig:example3}
\end{figure}

\begin{example}
\label{ex:trace-defn}
Let us illustrate the definitions and notations about traces
introduced in the previous paragraph. Consider the trace $\tr_3$ shown
in Fig.~\ref{fig:example3}. As in the introduction, we will refer to
the $i$th event in the trace by $e_i$. For trace $\tr_3$ we have ---
$\events{\tr_3} = \{e_1,e_2,\ldots e_{12}\}$; $\reads{\tr_3}
= \reads{\tr_3}(x) = \{e_7,e_{12}\}$; $\writes{\tr_3}
= \writes{\tr_3}(x) = \{e_2,e_5,e_9,e_{10}\}$. The last write of the
read events is as follows: $\lw{\tr_3}(e_7) = e_5$ and
$\lw{\tr_3}(e_{12}) = e_{10}$. The projection with respect to lock
$\lk$ is $\proj{\tr_3}{\lk} = e_1e_3e_4e_6$. The definition of
projection to a thread is subtle in the presence of forks and
joins. This can be seen by observing that $\proj{\tr_3}{t_4} =
e_8e_9e_{10}e_{11}$; this is because the fork event $e_8$ and the join
event $e_{11}$ are considered to be events of both threads $t_3$ and
$t_4$ by our definition. Finally, we illustrate $\ltho{\tr_3}(\cdot)$
through a few examples --- $\ltho{\tr_3}(e_2) = e_1$,
$\ltho{\tr_3}(e_7)$ is undefined, $\ltho{\tr_3}(e_9) = e_8$, and
$\ltho{\tr_3}(e_{11}) = e_{10}$. The cases of $e_9$ and $e_{11}$ are
the most interesting, and they follow from the fact that both $e_8$
and $e_{11}$ are also considered to be events of $t_4$.
\end{example}

\vspace*{0.1in}
\noindent
{\bf Orders.}  A given trace $\tr$ induces several total and partial
orders.  The total order
$\trord{\tr} \subseteq \events{\tr} \times \events{\tr}$, will be used
to denote the \emph{trace-order} --- $e \trord{\tr} e'$ iff either $e
= e'$ or $e$ appears before $e'$ in the sequence $\tr$.  Similarly,
the \emph{thread-order} is the smallest partial order
$\tho{\tr} \subseteq \events{\tr} \times \events{\tr}$ such that for
all pairs of events $e \trord{\tr} e'$ performed by the same thread,
we have $e \tho{\tr} e'$.
\begin{definition}[Happens-Before]
\label{def:hb}
Given trace $\tr$, the happens-before order $\hb{\tr}$ is the smallest
partial order on $\events{\tr}$ such that 
\begin{enumerate}[label=(\alph*)]
\item $\tho{\tr} \subseteq \hb{\tr}$, 
\item for every pair of events $e = \ev{t, \rel(\lk)}$, and,
$e' = \ev{t', \acq(\lk)}$ with
$e \trord{\tr} e'$, we have $e \hb{\tr} e'$
%
% Conditions for fork and join captured by thread order. Hence removed.
%
%\item for every pair of events $e
%= \ev{t,\fork(t')}$ and $e' = (t', op)$, $e \hb{\tr} e'$, and 
%\item for every pair of events $e = \ev{t', op}$ and $e' = \ev{t, \join(t')}$,
%$e \hb{\tr} e'$.
\end{enumerate}
\end{definition}

\begin{example}
\label{ex:hb}
We illustrate the definitions of $\trord{}, \tho{}$, and $\hb{}$ using
trace $\tr_3$ from Fig.~\ref{fig:example3}. Trace order is the
simplest; $e_i \trord{\tr_3} e_j$ iff $i \leq j$. Thread order is also
straightforward in most cases; the interesting cases of
$e_{8} \tho{\tr_3} e_9$ and $e_{10} \tho{\tr_3} e_{11}$ follow from
the fact that $e_8$ and $e_{11}$ are events of both threads $t_3$ and
$t_4$. Finally, let us consider $\hb{\tr_3}$. It is worth observing
that $e_7 \hb{\tr_3} e_9 \hb{\tr_3} e_{10} \hb{tr_3} e_{12}$ simply
because these events are thread ordered due to the fact that $e_8$ and
$e_{11}$ are events of both thread $t_3$ and $t_4$. In addition,
$e_2 \hb{\tr_3} e_5$ because $e_3 \hb{\tr_3} e_4$ by rule (b),
$e_2 \tho{\tr_3} e_3$ and $e_4 \tho{\tr_3} e_5$, and $\hb{\tr_3}$ is
transitive.
\end{example}

\vspace*{0.1in}
\noindent
{\bf Trace Reorderings.} Any trace of a concurrent program represents
one possible interleaving of concurrent events. The notion
of \emph{correct reordering}~\cite{cp2012,wcp2017} of trace $\tr$ identifies
all these other possible interleavings of $\tr$. In other words, if
$\tr'$ is a correct reordering of $\tr$ then any program that produces
$\tr$ may also produce $\tr'$. The definition of correct reordering is
given purely in terms of the trace $\tr$ and is agnostic of the
program that produced it. We give the formal definition below.
% \ucomment{We modify it}
% Correct reorderings have also been called \emph{feasible trace}
% in~\cite{rv2014}.
%
% \begin{definition}[Correct reordering]
% \label{def:correct-reorder}
% A trace $\tr'$ is said to be a correct reordering of
% a trace $\tr$ if
% \begin{enumerate*}[label=(\alph*)]
% %\item $\tr'$ is well formed, i.e., for every lock $\lk$, $\proj{\tr'}{\lk}$ 
% %  is a prefix of some string belonging to the regular language
% %  $(\cup_{t \in \threads{\tr'}} \ev{t, \acq(\lk)} \cdot 
% %  \ev{t, \rel(\lk)})^*$, 
% \item $\forall t \in \threads{\tr'}, \proj{\tr'}{t}$ is a prefix of 
%   $\proj{\tr}{t}$, and
% \item $\forall e \in \reads{\tr'}$, $\lw{\tr'}(e)$ exists iff
%   $\lw{\tr}(e)$ exists.  If it exists, then $\lw{\tr'}(e)
%   = \lw{\tr}(e)$.
% \end{enumerate*}
% \end{definition}
\begin{definition}[Correct reordering]
\label{def:correct-reorder}
A trace $\tr'$ is said to be a correct reordering of
a trace $\tr$ if
\begin{enumerate}[label=(\alph*)]
%\item $\tr'$ is well formed, i.e., for every lock $\lk$, $\proj{\tr'}{\lk}$ 
%  is a prefix of some string belonging to the regular language
%  $(\cup_{t \in \threads{\tr'}} \ev{t, \acq(\lk)} \cdot 
%  \ev{t, \rel(\lk)})^*$, 
\item $\forall t \in \threads{\tr'}, \proj{\tr'}{t}$ is a prefix of 
  $\proj{\tr}{t}$, and
\item for a read event $e = \ev{t, \rd(x)} \in \events{\tr'}$
such that $e$ is not the last event in $\proj{\tr'}{t}$,
$\lw{\tr'}(e)$ exists iff $\lw{\tr}(e)$ exists.  Further, if it
exists, then $\lw{\tr'}(e) = \lw{\tr}(e)$.
% \item \textcolor{green!80!black}{
% Let $e_1 = \ev{t_1, \rd(x)} \in \events{\tr}$
% and $e_2 = \ev{t_2, \wt(x)} \in \events{\tr}$ 
% be read and write events on the same variable $x$
% such that $e_1 \trord{\tr} e_2$.
% If $e_1 \in \events{\tr'}$, $e_2 \in \events{\tr'}$ 
% and $e_1$ is not the last event in $\proj{\tr'}{t_1}$,
% then $e_1 \trord{\tr'} e_2$.
% }
\end{enumerate}
\end{definition}
The intuition behind the above definition is the following.  A correct
reordering must preserve lock semantics (ensured by the fact that
$\tr'$ is a trace) and the order of events inside a given thread
(condition (a)).
% \ucomment{Because sequential consistency? Also because that's the best we can guess}.
Condition (b) captures \emph{local determinism}~\cite{rv2014}. 
That is, only the previous events in a given thread determine the next
event of the thread.
Since the underlying program that generated $\tr$ can have
branch events that depend upon the data in shared memory locations,
all reads in $\tr'$, except for the last events in each thread, 
must see the same value as in $\tr$; since our
traces don't record the value written, this can be ensured by
conservatively requiring every read to see the same write event.
If the last event of any thread in $\tr'$ is a read,
we allow that this event may not see the same value (and thus
the same last write event) as in $\tr$.
For example, consider the program and trace given in
Fig.~\ref{fig:example1}. The read event $\rd(y)$ in the conditional in
thread $t_2$ cannot be swapped with the preceding event $\wt(y)$ in
thread $t_1$, because that would result in a different branch being
taken in $t_2$, and the assignment $\texttt{x := 10}$ in $t_2$ will
never be executed. However, this is required only if the read event is
not the last event of the thread in the reordering. If it is the last
event, it does not matter what value is read, because it does not
affect future behavior. 

We note that the definition of correct reordering we have is more general
than in~\cite{wcp2017,cp2012} because of the relaxed assumption about
the last-write events corresponding to read events which are
not followed by any other events in their corresponding threads.
In other words, every correct reordering $\tr'$ of a trace $\tr$
according to the definition in~\cite{cp2012,wcp2017} is also a correct reordering of $\tr$
as per Definition~\ref{def:correct-reorder}, but the converse is not true.
On the other hand, the related notion of \emph{feasible} set of traces~\cite{rv2014}
allows for an even larger set of alternate reorderings that can be inferred
from an observed trace $\tr$ by only enforcing that 
the last-write event $e'$ corresponding to a read event $e$
must write the same value that $e$ reads in $\tr$. 
In particular, $e'$ may not be the same as $\lw{\tr}(e)$.

In addition to correct reorderings, another useful collection of
alternate interleavings of a trace is as follows. Under the assumption
that $\hb{\tr}$ identifies certain causal dependencies between events
of $\tr$, we consider interleavings of $\sigma$ that are consistent
with $\hb{\tr}$.
\begin{definition}[$\hb{}$-respecting trace]
\label{def:hb-respecting}
For trace $\tr$, we say trace $\tr'$ respects $\hb{\tr}$ if for any
$e,e' \in \events{\tr}$ such that $e \hb{\tr} e'$ and
$e' \in \events{\tr'}$, we have $e \in \events{\tr'}$ and
$e \trord{\tr'} e'$. 
%\ucomment{Say that $\events{\tr'} \subseteq \events{\tr}$}
\end{definition}
Thus, a $\hb{\tr}$-respecting trace is one whose events are downward
closed with respect to $\hb{\tr}$ and in which $\hb{\tr}$-ordered events are not
flipped. We will be using the above notion only when the trace $\tr'$
is a reordering of $\tr$, and hence $\events{\tr'} \subseteq \events{\tr}$.

\begin{example}
\label{ex:reorderings}
We give examples of correct reorderings of $\tr_3$ shown in
Fig.~\ref{fig:example3}. The traces $\rho_1 = e_1e_2e_7$, $\rho_2 =
e_4e_5e_6$, and $\rho_3 = e_1e_2e_3e_4e_5e_7$ are all examples of
correct reorderings of $\tr_3$. Among these, the trace $\rho_2$ is not
$\hb{}$-respecting because it is not $\hb{}$-downward closed ---
events $e_1,e_2,e_3$ are all HB-before $e_4$ and none of them are in
$\rho_2$.
\end{example}

\vspace*{0.1in}
\noindent
{\bf Race.} It is useful to recall the formal definition of a data
race, and to state the soundness guarantees of happens-before. Two
data access events $e = \ev{t_1, \mathsf{a}_1(x)}$ and $e'
= \ev{t_2,\mathsf{a}_2(x)}$ are said to be \emph{conflicting} if
$t_1 \neq t_2$, and at least one among $\mathsf{a}_1$ and
$\mathsf{a}_2$ is a write event. A trace $\tr$ is said to have a race
if it is of the form $\tr = \tr' e e' \tr''$ such that $e$ and $e'$
are conflicting; here $(e,e')$ is either called a race pair or a
race. A concurrent program is said to have a race if it has an
execution that has a race.

The partial order $\hb{\tr}$ is often employed for the purpose of
detecting races by analyzing program executions. In this context, it
is useful to define what we call an HB-race. A pair of conflicting events
$(e,e')$ is said to be an \emph{HB-race} if $e \trord{\tr} e'$ and $e$
and $e'$ are incomparable with respect to $\hb{\tr}$ (i.e., neither
$e \hb{\tr} e'$ nor $e' \hb{\tr} e$). We say an HB-race $(e,e')$ is
the \emph{first HB-race} if for any other HB-race $(f,f') \neq
(e,e')$ in $\tr$, either $e' \stricttrord{\tr} f'$, or $e' = f'$ and
$f \stricttrord{\tr} e$. 
For example, the pair $(e_2, e_3)$ in trace $\tr_1$ from Fig.~\ref{fig:example1} 
is the first HB-race of $\tr_1$.
The soundness guarantee of HB says that if a
trace $\tr$ has an HB-race, then the first HB-race is a valid data
race.
\begin{theorem}[Soundness of HB]
\label{thm:hb_sound}
Let $\tr$ be a trace with an HB-race, and let $(e, e')$ be the first
HB-race.  Then, there is a correct reordering $\tr'$ of $\tr$, such
that $\tr' = \tr'' e e'$.
\end{theorem}
Instead of sketching the proof of Theorem~\ref{thm:hb_sound}, we will
see that it follows from the main result of this paper, namely,
Theorem~\ref{thm:SHBSoundness}.

\begin{example}
\label{ex:hb-race}
We conclude this section by giving examples of HB-races. Consider
again $\tr_3$ from Fig.~\ref{fig:example3}. Among the different pairs
of conflicting events in $\tr_3$, the HB-races are $(e_2,e_7)$,
$(e_5,e_7)$, $(e_2,e_9)$, $(e_5,e_9)$, $(e_2,e_{10})$, $(e_5,e_{10})$,
$(e_2,e_{12})$, and $(e_5,e_{12})$.
\end{example}

\begin{remark}
% \ucomment{Feel free to edit, as needed.} 
Our model of executions and reorderings assume sequential consistency,
which is a standard model used by most race detection
tools. Executions in a more general memory model, such as 
Total Store Order (TSO),  would also have events that indicate when
 a local write was committed to the global
memory~\cite{HuangTSO16}. In that scenario,
the definition of correct reorderings would be similar, except that
``last write'' would be replaced by ``last observed write'', which
would either be the last committed write or the last write by the same
thread, whichever is later in the trace. The number of correct
reorderings to be considered would increase --- instead of just
considering executions where every write is immediately committed, as
we do here, we would also need to consider reorderings where the write
commits are delayed. However, since our results here are about proving
the existence of a reordered trace where a race is observed, they
carry over to the more general setting. We might miss race pairs that
could be shown to be in race in a weaker memory model, where more
reoderings are permitted, but the races we identify would still be
valid.
\end{remark}

\section{Characterizing Schedulable Races}
\label{sec:shb}
%!TEX root = main.tex

The example in Fig.~\ref{fig:example1} shows that not every HB-race
corresponds to an actual data race in the program. The goal of this
section is to characterize those HB-races which correspond to actual
data races. We do this by introducing a new partial order, called
schedulable happens-before, and using it to identify the actual data
races amongst the HB-races of a trace. We begin by characterizing the
HB-races that correspond to actual data races.
\begin{definition}[$\hb{\tr}$-schedulable race]
\label{def:hb-sched-race}
Let $\tr$ be a trace and let $e \trord{\tr} e'$ be conflicting events in $\tr$.
We say that $(e,e')$ is a $\hb{\tr}$-schedulable race if there is a correct
reordering $\tr'$ of $\tr$ that respects $\hb{\tr}$ and $\tr' = \tr''
ee'$ or $\tr' = \tr'' e'e$ for some trace $\tr''$. 
%
% We should keep this definition more general than needed.
% \ucomment{Only $\tr''ee'$}
\end{definition}

Note that any $\hb{\tr}$-schedulable race is a valid data race in
$\tr$. Our aim is to characterize $\hb{\tr}$-schedulable races by means
of a new partial order. 
The new partial order, given below, is a strengthening of $\hb{}$.
\begin{definition}[Schedulable Happens-Before]
\label{def:shb}
Let $\tr$ be a trace. Schedulable happens-before, denoted by
$\shb{\tr}$, is the smallest partial order on $\events{\tr}$ such that
\begin{enumerate}[label=({\alph*})]
\item $\hb{\tr} \subseteq \shb{\tr}$
\item $\forall e,e' \in \events{\tr}, e' \in \reads{\tr} \land 
  e = \lw{\sigma}(e') \implies e \shb{\tr} e'$
\end{enumerate}
\end{definition}
The partial order $\shb{\tr}$ can be used to characterize
$\hb{\tr}$-schedulable races. We state this result, before giving
examples illustrating the definition of $\shb{\tr}$.
\begin{theorem}%[SHB Soundness]
\label{thm:SHBSoundness}
Let $\tr$ be a trace and $e_1 \trord{\tr} e_2$ be conflicting events in $\tr$.
$(e_1, e_2)$ is an $\hb{\tr}$-schedulable race iff 
either $\ltho{\tr}(e_2)$ is undefined, or $e_1 \not\leq^\tr_{\mathsf{SHB}} \ltho{\tr}(e_2)$.
  % $\neg (e_1 \shb{\tr} \ltho{\tr}(e_2))$.
% there is no event $e \in \events{\tr} \setminus \{e_1, e_2\}$ such
%   that $e_1 \shb{\tr} e \shb{\tr} e_2$.
\end{theorem}

\begin{proof}(Sketch)
The full proof is presented in Appendix~\ref{app:shb-proof}; here we
sketch the main ideas.  We observe that if $\tr'$ is a correct
reordering of $\tr$ that also respects $\hb{\tr}$, then $\tr'$ also
respects $\shb{\tr}$ except possibly for the last events of every
thread in $\tr'$. That is, for any $e,e'$ such that $e \shb{\tr} e'$,
$e' \in \events{\tr'}$, and $e'$ is not the last event of some thread
in $\tr'$, we have $e \in \events{\tr'}$ and $e \trord{\tr'}
e'$. Therefore, if $e \shb{\tr} \ltho{\tr}(e_2)$, then any correct
reordering $\tr'$ respecting $\hb{\tr}$ that contains both $e_1$ and
$e_2$ will also have $e = \ltho{\tr}(e_2)$. Further since $e$ is not
the last event of its thread (since $e_2$ is present in $\tr'$) and
$e_1 \shb{\tr} e$, $e$ must occur between $e_1$ and $e_2$ in
$\tr'$. Therefore $(e_1,e_2)$ is not a $\hb{\tr}$-schedulable race.
The other direction can be established as follows. Let $\tr''$ be the
trace consisting of events that are $\shb{\tr}$-before $e_1$ or
$\ltho{\tr}(e_2)$ (if defined), ordered as in $\tr$. Define
$\tr'= \tr''e_1e_2$.  We prove that when $e_1$ and $e_2$ satisfy the
condition in the theorem, $\tr'$ as defined here, is a correct
reordering and also respects $\hb{\tr}$.
\end{proof}

\begin{figure}
\centering
\execution{4}{
  \figev{1}{\acq(\lk)}
  \figev{1}{\wt(x)}
  \figev{2}{\rd(x)}
  \figev{2}{\wt(y)}
  \figev{2}{\wt(x)}
  \figev{1}{\rd(x)}
  \figev{1}{\rel(\lk)}
  \figev{4}{\acq(\lk)}
  \figev{4}{\wt(z)}
  \figev{3}{\rd(z)}
  \figev{3}{\wt(y)}
  \figev{3}{\wt(z)}
  \figev{4}{\rd(z)}
  \figev{4}{\rel(\lk)}
}
\caption{Trace $\tr_4$.}
\label{fig:example4}
\end{figure}

We now illustrate the use of $\shb{}$ through some examples.
\begin{example}
In this example, we will look at different traces, and see how
$\shb{}$ reasons. Like in the introduction, we will use $e_i$ to refer
to the $i$th event of a given trace (which will be clear from
context). Let us begin by considering the example program and trace
$\tr_1$ from Fig.~\ref{fig:example1}. Notice that $\hb{\tr_1}
= \tho{\tr_1}$, and so $(e_1,e_4)$ and $(e_2,e_3)$ are
HB-races. Because $e_2 = \lw{\tr_1}(e_3)$, we have $e_1 \shb{\tr_1}
e_2 \shb{\tr_1} e_3 \shb{\tr_1} e_4$. Using
Theorem~\ref{thm:SHBSoundness}, we can conclude correctly that (a)
$(e_2,e_3)$ is $\hb{\tr_1}$-schedulable as $\ltho{\tr_1}(e_3)$ is
undefined, but (b) $(e_1,e_4)$ is not, as
$e_1 \shb{\tr} \ltho{\tr}(e_4) = e_3$.

Let us now consider trace $\tr_2$ from
Fig.~\ref{fig:example2}. Observe that $\hb{\tr_2} = \shb{\tr_2}
= \tho{\tr_2}$, and so both $(e_1,e_4)$ and $(e_2,e_3)$ are
$\hb{\tr_2}$-schedulable races by Theorem~\ref{thm:SHBSoundness}. Note
that, unlike force ordering, $\shb{\tr_2}$ correctly identifies all
real data races.

Finally, let us consider two trace examples that highlight the kind of
subtle reasoning $\shb{}$ is capable of. Let us begin with $\tr_3$
from Fig.~\ref{fig:example3}. As observed in Example~\ref{ex:hb-race},
the only HB-races in this trace are $(e_2,e_7)$, $(e_5,e_7)$,
$(e_2,e_9)$, $(e_5,e_9)$, $(e_2,e_{10})$, $(e_5,e_{10})$,
$(e_2,e_{12})$, and $(e_5,e_{12})$. Both $(e_2,e_7)$ and $(e_5,e_7)$
are $\hb{\tr_3}$-schedulable as demonstrated by the reorderings
$\rho_1$ and $\rho_3$ from Example~\ref{ex:reorderings}. However, the
remaining are not real data races. Let us consider the pairs
$(e_2,e_9)$ and $(e_5,e_9)$ for
example. Theorem~\ref{thm:SHBSoundness}'s justification for it is as
follows: $e_2 \hb{\tr_3} e_5 = \lw{\tr_3}(e_7) \tho{\tr_3} e_8
= \ltho{\tr_3}(e_9)$. But, let us unravel the reasoning behind why
neither $(e_2,e_9)$ nor $(e_5,e_9)$ are data races. Consider an
arbitary correct reordering $\tr'$ of $\tr_3$ that respects
$\hb{\tr_3}$ and contains $e_9$. Since $e_8$ is also an event of
$t_4$, $e_8 \in \events{\tr'}$. In addition, $e_7 \in \events{\tr'}$
as $e_7 \tho{\tr_3} e_8$. Now, since $e_5 = \lw{\tr_3}(e_7)$, $e_5$ is
before $e_7$ in $\tr'$ and since $e_2 \hb{\tr_3} e_5$, $e_2$ must also
be before $e_7$. Therefore, $e_7$ and $e_8$ will be between $e_2$ and
$e_9$ and between $e_5$ and $e_9$. Similar reasoning can be used to
conclude that the other pairs are not $\hb{\tr_3}$-schedulable as
well. 

Lastly, consider trace $\tr_4$ shown in Fig.~\ref{fig:example4}. In
this case, $\shb{\tr_4} = \trord{\tr_4}$. All conflicting memory
accesses are in HB-race.  While HB correctly identifies the first race
$(e_2, e_3)$ as valid, there are 3 HB-races that are not real data
races --- $(e_2,e_5)$, $(e_9,e_{12})$, and $(e_4,e_{11})$. $(e_2,e_5)$
is not valid because any correct reordering of $\tr_4$ must have $e_2$
before $e_3$ and $e_3$ before $e_5$. This is also captured by SHB
reasoning because $e_2 \shb{\tr_4} e_3 \tho{\tr_4} e_4
= \ltho{\tr_4}(e_5)$. A similar reasoning shows that $(e_9,e_{12})$ is
not valid. The interesting case is that of $(e_4,e_{11})$. Here, in
any correct reordering $\tr'$ of $\tr_4$, the following must be true:
(a) if $e_4 \in \events{\tr'}$ then $e_1 \in \events{\tr'}$; (b) if
$e_{11} \in \events{\tr'}$ then $e_8 \in \events{\tr'}$; (c) if
$\{e_1,e_4,e_7\} \subseteq \events{\tr'}$ then $e_1 \trord{\tr'}
e_4 \trord{\tr'} e_7$; and (d) if
$\{e_8,e_{11},e_{14}\} \subseteq \events{\tr'}$ then $e_8 \trord{\tr'}
e_{11} \trord{\tr'} e_{14}$. Therefore, any correct reordering $\tr'$
of $\tr_4$ containing both $e_4$ and $e_{11}$ contains $e_1$ and $e_8$
(because of (a) and (b)) and must contain at least one of $e_7$ or
$e_{14}$ to ensure that critical sections of $\lk$ don't overlap. Then
in $\tr'$, $e_4$ and $e_{11}$ cannot be consecutive because either
$e_7$ or $e_{14}$ will appear between them (properties (c) and
(d)). This is captured using SHB and Theorem~\ref{thm:SHBSoundness} by
the fact that $e_4 \shb{\tr_3} e_7 \shb{\tr_3} e_{10}
= \ltho{\tr_4}(e_{11})$.
\end{example}

We conclude this section by observing that the soundness guarantees of
HB (Theorem~\ref{thm:hb_sound}) follows from
Theorem~\ref{thm:SHBSoundness}. Consider a trace $\tr$ whose first
HB-race is $(e_1,e_2)$. We claim that $(e_1,e_2)$ is a
$\hb{\tr}$-schedulable race. Suppose (for contradiction) it is
not. Then by Theorem~\ref{thm:SHBSoundness}, $e = \ltho{\tr}(e_2)$ is
defined and $e_1 \shb{\tr} e$.  Now observe that we must have
$\neg(e_1 \hb{\tr} e)$ (or otherwise $e_1 \hb{\tr} e_2$, contradicting
our assumption that $(e_1, e_2)$ is an HB-race).  Then, by the
definition of $\shb{\tr}$ (Definition~\ref{def:shb}), there are two
events $e_3$ and $e_4$ (possibly same as $e_1$ and $e$) such that
$e_1 \shb{\tr} e_3$, $e_3 = \lw{\tr}(e_4)$, $e_4 \shb{\tr} e$, and
$\neg (e_3 \hb{\tr} e_4)$.  Then $(e_3,e_4)$ is an HB-race, and it
contradicts the assumption that $(e_1,e_2)$ is the first HB-race.

The above argument that Theorem~\ref{thm:hb_sound} follows from
Theorem~\ref{thm:SHBSoundness}, establishes that our SHB-based
analysis using Theorem~\ref{thm:SHBSoundness} does not miss the race
detected by a \emph{sound} HB-based race detection algorithm.

\section{Algorithm for Detecting \texorpdfstring{$\hb{}$}{HB}-schedulable Races}
\label{sec:algo}
%!TEX root = main.tex

We will discuss two algorithms for detecting races identified by the
$\shb{}$ partial order. 
The algorithm is based on efficient, vector clock based
computation of the $\shb{}$-partial order. It is similar to the
standard \textsc{Djit}$^+$
algorithm~\cite{Pozniansky:2003:EOD:966049.781529} to detect
HB-races. We will first briefly discuss vector clocks and
associated notations.  Then, we will discuss a one-pass streaming
vector clock algorithm to compute $\shb{}$ for detecting races.  
Finally, we will discuss how epoch optimizations, similar to \fasttrack~\cite{fasttrack} 
can be readily applied in our setting to enhance performance
of the proposed vector clock algorithm.

% \vspace{-0.1in}
\subsection{Vector Clocks and Times}
% \vspace{-0.1in}
% \paragraph{Vector Times and clocks} 
A vector \emph{time} \textit{VT : $\threads{\tr}$ $\to$ Nat} maps each
thread in a trace $\tr$ to a natural number.  Vector times support
comparison operation $\cle$ for point-wise comparison, join operation
($\mx$) for point-wise maximum, and update operation $V[n/t]$ which
assigns the time $n \in $\textit{Nat} to the component $t
\in \threads{\tr}$ in the vector time $V$.  Vector time $\bot$ maps all
threads to 0. Formally, 
% \vspace*{.2\baselineskip} 
% \begin{tabular}{rclcrcl}
% $V_1 \cle V_2$ & \! {iff} \! & $\forall t: V_1(t) \le V_2(t)$ &  & $V_1 \mx V_2$ & = & $\lambda t: \mathit{max}(V_1(t), V_2(t))$\\
% $V[n/u]$ & = & $\lambda t: \mathtt{if}\; (t = u)\; \mathtt{then}\; n\; \mathtt{else}\; V(t)$ & & $\bot$ & = & $\lambda t: 0$\\
% \end{tabular}
%
% \vspace*{.2\baselineskip} 
\begin{align*}
\begin{array}{rclr}
V_1 \cle V_2 & \text{iff} & \forall t: V_1(t) \le V_2(t) & \text{(Point-wise Comparison)}\\
V_1 \mx V_2 & = & \lambda t: \mathit{max}(V_1(t), V_2(t))  & \text{(VC-Join)}\\
V[n/u] & = & \lambda t: \mathtt{if}\; (t = u)\; \mathtt{then}\; n\; \mathtt{else}\; V(t) & \text{(VC-Update)}\\
\bot & = & \lambda t: 0 & \text{(VC-Bottom)}
\end{array}
\end{align*}
% \noindent 
Vector \emph{clocks} are place holders for vector timestamps, or
variables whose domain is the space of vector times.
All the above operations, therefore, also apply to vector clocks.
The algorithms described next maintain a state comprising of
several vector  clocks, whose values, at specific instants, 
will be used to assign timestamps to events.
We will use double struck font ($\Cc$, $\Ll$, $\Rr$, etc.,) for vector clocks 
and normal font ($C$, $R$, etc.,) for vector times .

%!TEX root = main.tex

\subsection{Vector Clock Algorithm for Detecting Schedulable Races}

Algorithm~\ref{algo:vc} depicts the vector clock algorithm for
detecting $\hb{}$-schedulable races using the $\shb{}$ partial order.
Similar to the vector clock algorithm for detecting HB races,
Algorithm~\ref{algo:vc} maintains a state comprising of several vector
clocks.  The idea behind Algorithm~\ref{algo:vc} is to use these
vector clocks to \emph{assign a vector timestamp} to each event $e$
(denoted by $C_e$) such that the ordering relation on the assigned
timestamps ($\cle$) enables determining the partial order $\shb{}$ on
events.  This is formalized in Theorem~\ref{thm:isomorphicVC}.  The
algorithm runs in a streaming fashion and processes each event in the
order in which it occurs in the trace.  Depending upon the type of the
observed event, an appropriate handler is invoked.  The formal
parameter $t$ in each of the handlers refers to the thread performing
the event, and the parameters $\lk$, $x$ and $u$ represent the lock
being acquired or released, the memory location being accessed and the
thread being forked or joined, respectively.  The procedure
\texttt{Initialization} assigns the initial values to the vector
clocks in the state.  We next present details of different parts of
the algorithm.

%!TEX root = main.tex
% \begin{Centering}
\SetKwProg{myalg}{procedure}{}{}

\begin{algorithm*}
\begin{multicols}{2}
\SetKwFunction{facq}{acquire}
\SetKwFunction{frel}{release}
\SetKwFunction{fork}{fork}
\SetKwFunction{join}{join}
\SetKwFunction{read}{read}
\SetKwFunction{wr}{write}
\SetKwFunction{init}{Initialization}

\myalg{\init}{
    \lForEach{$t$}{
    $\Cc_t$ := $\bot[1/t]$
    } 
    \lForEach{$\ell$}{
    $\Ll_\ell$ := $\bot$ % $\lambda t: 0$
    } 
    \For{$x \in \vars{}$}{ 
    $\LW_x$ := $\bot$; \\ % $\lambda t: 0$ 
    $\Rr_x$ := $\bot$; \\% $\lambda t: 0$ ; 
    $\Ww_x$ := $\bot$;  % $\lambda t: 0$ ;
    }
}

\myalg{\facq{$t$, $\ell$}}{ %
    % \If{previous event in $t$ was $\rel$, $\wt$ or $\fork$}{
    %     $\Cc_t(t)$ := $\Cc_t(t) + 1$ ;
    % }
    $\Cc_t$ := $\Cc_t \mx \Ll_\ell$ ;
}

\myalg{\frel{$t$, $\ell$}}{
    $\Ll_\ell$ := $\Cc_t$ ; \\
    $\Cc_t(t)$ := $\Cc_t(t) + 1$ ; {\scriptsize \texttt{(* next event *)}}
}

\myalg{\fork{$t$, $u$}}{
    $\Cc_u$ := $\Cc_t[1/u]$ ; \\
    $\Cc_t(t)$ := $\Cc_t(t) + 1$ ; {\scriptsize \texttt{(* next event *)}}
}

\myalg{\join{$t$, $u$}}{
    $\Cc_t$ := $\Cc_t \mx \Cc_u$ ;
}

\myalg{\read{$t$, $x$}}{
    \If{$\neg (\Ww_x \cle \Cc_t)$}{
      declare `\textbf{race with write}';
    }
    $\Cc_t$ := $\Cc_t \mx \LW_x$; \\
    $\Rr_x(t)$ := $\Cc_t(t)$; \\
}

\myalg{\wr{$t$, $x$}}{
    \If{$\neg (\Rr_x \cle \Cc_t)$}{
      declare `\textbf{race with read}'; \\
    }
    \If{$\neg (\Ww_x \cle \Cc_t)$}{
      declare `\textbf{race with write}'; \\
    } 
    $\LW_x$ := $\Cc_t$ ; \\
    $\Ww_x(t)$ := $\Cc_t(t)$; \\
    $\Cc_t(t)$ := $\Cc_t(t) + 1$ ; {\scriptsize \texttt{(* next event *)}}
}
\end{multicols}
\vspace*{0.25cm}
\caption{\textit{Vector Clock for Checking $\shb{}$-schedulable races}}
\label{algo:vc}
\end{algorithm*}

\subsubsection{Vector clocks in the state}
\label{subsec:vc}
The description of each of the vector clocks that are
maintained in the state of Algorithm~\ref{algo:vc} is
as follows:

\begin{enumerate}
\item \textbf{Clocks} $\Cc_t$: For every thread $t$ in the trace
being analyzed, the algorithm maintains a vector clock $\Cc_t$.
At any point during the algorithm, 
let us denote by $C_{e_t}$ the last event performed by thread
$t$ in the trace so far.
Then, the \emph{timestamp} $C_{e_t}$ 
of the event $e_t$ can be obtained from the value of the clock $\Cc_t$ as follows.
If $e_t$ is a read, acquire or a join event, then $C_{e_t} = \Cc_t$, 
otherwise $C_{e_t} = \Cc_t[(c-1)/t]$, where $c = \Cc_t(t)$.
% Loosely speaking, at any point during the algorithm, the clock $\Cc_t$ stores the
% timestamp $C_{e_t}$ of the last event $e_t$ seen in thread $t$. 
% \ucomment{These lines are marked with a $\dagger$ sign in the algorithm.}

\item \textbf{Clocks} $\Ll_\lk$: The algorithm maintains
a vector clock $\Ll_\lk$ for every lock $\lk$ in the trace.
At any point during the algorithm, the clock $\Ll_\lk$ stores the
timestamp $C_{e_\lk}$, where $e_\lk$ is the last event of the form
$e_\lk = \ev{\cdot, \rel(\lk)}$, in the trace seen so far. 
% Intuitively, the clock $\Ll_\lk$ is used as a \emph{channel} 
% to communicate the timestamps of the release events of lock $\lk$.

\item \textbf{Clocks} $\LW_x$: For every memory location $x$
accessed in the trace, the algorithm maintains a clock $\LW_x$ ($\mathbb{L}$ast $\mathbb{W}$rite to $x$)
to store the timestamp $C_{e_x}$, of the last event $e_x$ 
of the form $\ev{\cdot, \wt(x)}$.

\item \textbf{Clocks} $\Rr_x$ and $\Ww_x$: The clocks
$\Rr_x$ and $\Ww_x$ store the read and write \emph{access histories} of
each memory location $x$.
% Formally, for a trace $\tr$ processed by the algorithm,
At any point in the algorithm,
the vector time $R_x$ stored in the 
the $\mathbb{R}$ead access history clock $\Rr_x$ is
such that $\forall t, R_x(t) = C_{e_t^{\rd(x)}}(t)$ where 
$e_t^{\rd(x)}$ is the last event of thread $t$ that reads $x$ in 
the trace seen so far.
Similarly, the vector time $W_x$ stored in the $\mathbb{W}$rite access history
clock $\Ww_x$ is such that $\forall t, W_x(t) = C_{e_t^{\wt(x)}}(t)$ where $e_t^{\wt(x)}$ 
is the last event of thread $t$ that writes to $x$ in 
the trace seen so far.
\end{enumerate}

The clocks $\Cc_t$, $\Ll_\lk$, $\LW_x$ are used to correctly compute the
timestamps of the events, while the access history clocks $\Rr_x$ and $\Ww_x$
are used to detect races.

\subsubsection{Initialization and Clock Updates}

For every thread $t$, the clock $\Cc_t$ 
is initialized to the vector time $\bot[1/t]$.
Each of the clocks $\Ll_\lk$, $\LW_x$, $\Rr_x$ and $\Ww_x$ are initialized to $\bot$. 
This is in accordance with the semantics of these clocks presented 
in Section~\ref{subsec:vc}.

When processing an acquire event $e = \ev{t, \acq(\lk)}$, the algorithm
reads the clock $\Ll_\lk$ and updates the clock $\Cc_t$ 
with $\Cc_t \sqcup \Ll_\lk$ (see Line 9).
This ensures that the timestamp $C_e$ (
which is the value of the clock $\Cc_t$ after executing Line 9) 
is such that $C_{e'} \cle C_e$ for every $\lk$-release event $e' = \ev{t', \rel(\lk)}$
observed in the trace so far.

At a release event $e = \ev{t, \rel(\lk)}$,
the algorithm writes the timestamp
$C_e$ of the current event $e$
to the clock $\Ll_\lk$ (see Line 11).
Notice that $e$ is also
the last release event of lock $\lk$ in the trace seen so far,
and thus, this update correctly maintains the invariant
stated in Section~\ref{subsec:vc}.
This update ensures that any future events that 
acquire the lock $\lk$ can update their timestamps correctly.
The algorithm then increments the local clock $\Cc_t(t)$ (Line 12).
This ensures that if the next event $e'$ in the thread $t$ and the 
next acquire event $f$ of lock $\lk$ satisfy $e' \not\leq_\textsf{SHB} f$,
then the timestamps of these events satisfy $C_{e'} \not\cle C_f$.
This is crucial for the correctness of the algorithm (Theorem~\ref{thm:isomorphicVC}).

The updates performed by the algorithm at a fork (resp. join)
event are similar to the updates performed when observing a release (resp. acquire) event.
The update at Line 14 is equivalent to the update 
$\Cc_u$ := $\Cc_t \sqcup \Cc_u$ and ensures that the
timestamp of each event $e' = \ev{u, \cdot}$ performed by the forked thread $u$
satisfy $C_e \cle C_{e'}$, where $e$ is the current event forking the new thread $u$.
Similarly, the update performed at Line 17 when processing the join
event $e = \ev{t, \join(u)}$ ensures that the timestamp of
each event $e' = \ev{u, \cdot}$ of the joined thread $u$
is such that $C_{e'} \cle C_e$.

At a read event $e = \ev{t, \rd(x)}$, the clock $\Cc_t$ is updated with 
the join $\Cc_t \sqcup \LW_x$ (Line 21).
Recall that $\LW_x$ stores the timestamp of 
the last event that writes to $x$ (or in other words, the event $\lw{}(e)$)
in the trace seen so far.
This ensures that the timestamps $C_e$ and $C_{\lw{}(e)}$
satisfy $C_{\lw{}(e)} \cle C_e$.
In addition, the algorithm also updates the component $\Rr_x(t)$ with the 
local component of the clock
$\Cc_x$ (Line 22) in order to maintain the invariant described in Section~\ref{subsec:vc}.

At a write event $e = \ev{t, \wt(x)}$, the algorithm updates the value of the
last-write clock $\LW_x$ (Line 28) with the timestamp $C_e$ stored in $\Cc_t$.
The component $\Ww_x(t)$ is updated with the value of the local component $\Cc_t(t)$
to ensure the invariant described in Section~\ref{subsec:vc} is maintained correctly.
Finally, similar to the increment after a release event, the local clock is incremented
in Line 30.

\subsubsection{Checking for races}

At a read/write event $e$, the algorithm determines if
there is a conflicting event $e'$ in the trace seen so far such that
$(e', e)$ is an $\hb{}$-schedulable race.
From Theorem~\ref{thm:SHBSoundness} and Theorem~\ref{thm:isomorphicVC}, 
it follows that it is sufficient to check if $C_{e'} \not\cle C_{\ltho{}(e)}$.
However, since the algorithm does not explicitly store the
timestamps of events, we use the access histories $\Rr_x$
and $\Ww_x$ to check for races.
Below we briefly describe these checks.
The formal statement of correctness is presented in Theorem~\ref{thm:correct-races}
and its proof is presented in Appendix~\ref{app:algo}.
We briefly outline the ideas here.

Recall that, for an event $e = \ev{t, \cdot}$ 
if $\ltho{}(e)$ is undefined, the \texttt{Initialization}
procedure ensures that  $C_e = \bot[1/t]$.
In this case, we have $V \not\cle C_e$, for any
vector-timestamp $V$ with non-negative entries
such that $V(t) = 0$, $\bot \cle V$ and $V \neq \bot$. 
Algorithm~\ref{algo:vc} correctly reports a race in 
this case (see Lines 19-20, 24-27).

On the other hand, if $\ltho{}(e)$ is defined,
then the clock $\Cc_t$, at Line 19, 24 or 26, 
is either the timestamp $C_{\ltho{}(e)}$ (if $\ltho{}(e)$
was a read, join or an acquire event) or the timestamp 
$C_{\ltho{}(e)}[(c+1)/t]$, where $c = C_{\ltho{}(e)}(t)$
(if $\ltho{}(e)$
was a write, fork or a release event).
In either case, if the check $\Ww_x \cle \Cc_t$ at Line 19 fails,
then the read event $e$ being processed is correctly declared to
be in race with an earlier conflicting write event.
Similarly, Algorithm~\ref{algo:vc} reports that
a write event $e$ is in race with an earlier read (resp. write)
event based on whether the check on Line 24 (resp. Line 26) fails or not. 

\subsubsection{Correctness and Complexity}

Here, we fix a trace $\tr$.  Recall that, for an event $e$, we say
that $C_e$ is the timestamp assigned by Algorithm~\ref{algo:vc} to
event $e$.  Theorem~\ref{thm:isomorphicVC} asserts that the time
stamps computed by Algorithm~\ref{algo:vc} can be used to determine
the partial order $\shb{\tr}$.

\begin{theorem}
\label{thm:isomorphicVC}
For events $e, e' \in \events{\tr}$ such that $e \trord{\tr} e'$,
$C_e \cle C_{e'}$ iff $e \shb{\tr} e'$
\end{theorem}

Next, we state the correctness of the algorithm.
We say that Algorithm~\ref{algo:vc} reports a race
at an event $e$, if it executes lines 20, 25 or 27
while processing the handler corresponding to $e$.

\begin{theorem}
\label{thm:correct-races}
Let $e$ be a read/write event $e \in \events{\tr}$.
Algorithm~\ref{algo:vc} reports a race at $e$
iff there is an event $e' \in \events{\tr}$ such that $(e', e)$
is an $\hb{\tr}$-schedulable race.
\end{theorem}

The following theorem states that the asymptotic 
time and space requirements for Algorithm~\ref{algo:vc}
are the same as that of the standard HB algorithm.

\begin{theorem}
\label{thm:complexityVC}
For a trace $\tr$ with $n$ events, $T$ threads, $V$ variables, and $L$
locks, Algorithm~\ref{algo:vc} runs in time $O(nT\log n)$ and uses
$O((V+L+T)T\log n)$ space.
\end{theorem}

The proofs of Theorem~\ref{thm:isomorphicVC}, Theorem~\ref{thm:correct-races} and
Theorem~\ref{thm:complexityVC} are presented in Appendix~\ref{app:algo}.

\subsubsection{Differences from the HB algorithm}
While the spirit of Algorithm~\ref{algo:vc} is similar to standard HB
vector clock algorithms (such as \textsc{Djit}$^+$
\cite{Pozniansky:2003:EOD:966049.781529}), it differs from them in the
following ways.  First, we maintain an additional vector clock $\LW_x$
to track the timestamp of the last event that writes to memory
location $x$ (line 28), and use this clock to correctly update $\Cc_t$
(line 21).  This difference is a direct consequence of the additional
ordering edges in the $\shb{}$ partial order---every read event $e$ is
ordered after the event $\lw{}(e)$, unlike $\hb{}$.  Second, the
`local' component of the clock $\Cc_t$ is also incremented after every
write event (line 19), in addition to after a release or a fork event
(in contrast with \textsc{Djit}$^+$).  This is to ensure correctness
in the following scenario.  Let $e, e'$ and $e''$ be events such that
$e = \ev{t, \rd(x)}\in \reads{}$, $e' = \ev{t', \wt(x)} = \lw{}(e)$
($t'$ may be different from $t$), and $e''$ is the next event after
$e'$ in the thread $t'$.  Incrementing the local component of the
clock $\Cc_{t'}$ ensures that the vector timestamps of $e$ and $e''$
are ordered only when $e'' \shb{} e$.  Third, our algorithm remains
sound even beyond the first race, in contrast to \textsc{Djit}$^+$,
which can lead to false positives beyond the first race.% reported.

\subsection{Epoch optimization}
\label{sec:algo_epoch}

The epoch optimization, popularized
by \textsc{FastTrack}~\cite{fasttrack} exploits the insight that
`\textit{the full generality of vector clocks is unnecessary in most
cases}', and can result in significant performance enhancement,
especially when the traces are predominated by read and write events.

An epoch is a pair of an integer $c$ and a thread $t$, denoted by $c@t$.
Intuitively, epoch $c@t$ can be treated as the vector time $\bot[c/t]$.
Thus, in order to compare an epoch $c@t$ with vector time $V$, 
it suffices to compare the $t$-th component of $V$ with $c$.
That is,
% \begin{tabular}{rcl}
% $c@t \cle V$ & iff & $c \leq V(t)$.
% \end{tabular}
\begin{align*}
\begin{array}{rcl}
c@t \cle V & \text{iff} & c \leq V(t).
\end{array}
\end{align*}
Therefore, comparison between epochs is less expensive than that
between vector times --- $O(1)$ as opposed to $O((|\threads{\tr}|)$
for full vector times. To exploit this speedup, some vector clocks in
the new algorithm will adaptively store either epochs or vector times.

%!TEX root = main.tex
% \begin{Centering}
\SetKwProg{myalg}{procedure}{}{}

\begin{algorithm*}
\begin{multicols}{2}
\SetKwFunction{facq}{acquire}
\SetKwFunction{frel}{release}
\SetKwFunction{fork}{fork}
\SetKwFunction{join}{join}
\SetKwFunction{read}{read}
\SetKwFunction{wr}{write}
\SetKwFunction{init}{Initialization}

% \myalg{\init}{
%     \lForEach{$t$}{
%     $\Cc_t$ := $\bot[1/t]$
%     } 
%     \lForEach{$\ell$}{
%     $\Ll_\ell$ := $\bot$ % $\lambda t: 0$
%     } 
%     \For{$x \in \vars{}$}{ 
%     $\LW_x$ := $\bot$; \\ % $\lambda t: 0$ 
%     $\Rr_x$ := $0@0$; \\% $\lambda t: 0$ ; 
%     $\Ww_x$ := $0@0$;  % $\lambda t: 0$ ;
%     }
% }

\myalg{\read{$t$, $x$}}{
    \If{$\neg (\Ww_x \cle \Cc_t)$}{
      declare `\textbf{race with write}';
    }
    $\Cc_t$ := $\Cc_t \mx \LW_x$; \\
   \eIf{$\Rr_x$ is an epoch $c@u$}{
		\eIf{$c \leq \Cc_t(u)$}{
			$\Rr_x$ := $\Cc_t(t)@t$;
		}{
			$\Rr_x$ := $\bot[\Cc_t(t)/t][c/u]$;
		}
	}{
		$\Rr_x(t)$ := $\Cc_t(t)$;
	}
}

\myalg{\wr{$t$, $x$}}{
    \If{$\neg (\Rr_x \cle \Cc_t)$}{
      declare `\textbf{race with read}'; \\
    }
    % \If{$\neg (\Ww_x \cle \Cc_t)$}{
    %   declare `\textbf{race with write}'; \\
    % } 
    \eIf{$(\Ww_x \cle \Cc_t)$}{
		$\Ww_x$ := $\Cc_t(t)@t$
   	}{
   		declare `\textbf{race with write}'; \\	
   		\eIf{$\Ww_x$ is an epoch $c@u$}{
			$\Ww_x$ := $\bot[\Cc_t(t)/t][c/u]$;
   		}{
   			$\Ww_x(t)$ := $\Cc_t(t)$;
   		}

   	}
    $\LW_x$ := $\Cc_t$ ; \\
    $\Ww_x(t)$ := $\Cc_t(t)$; \\
    $\Cc_t(t)$ := $\Cc_t(t) + 1$ ; {\scriptsize \texttt{(* next event *)}}
}
\end{multicols}
\vspace*{0.25cm}
\caption{\textit{Epoch Optimization for Algorithm~\ref{algo:vc}}}
\label{algo:epoch}
\end{algorithm*}

Algorithm~\ref{algo:epoch} applies the epoch optimization to
Algorithm~\ref{algo:vc}.  Here, similar to the \fasttrack~algorithm,
we allow clocks $\Rr_x$ and $\Ww_x$ to be adaptive, while other clocks
($\Cc_t$, $\Ll_\lk$ and $\LW_x$) always store vector times.  The
optimization only applies to the $\read$ and $\texttt{write}$ handlers
and thus we omit the other handlers from the description as they are
same as those described in Algorithm~\ref{algo:vc}.  We also omit
the \texttt{Initialization} procedure which only differs in that the
$\Rr_x$ and $\Ww_x$ are initialized to the epoch $0@0$.

Depending upon how these clocks compare with the thread's clock $\Cc_t$,
the clocks switch back and forth between epoch and vector time values:
\begin{itemize}
\item Initially, both $\Rr_x$ and $\Ww_x$ are assigned the epoch $0@0$.
The element $0@0$ can be thought of as the analogue of $\bot$.

\item
The clock $\Ww_x$ is \emph{fully adaptive} --- it can switch back
and forth between vector and epoch times depending upon how it compares with $\Cc_t$. 
Notice that, in the \fasttrack~algorithm proposed in~\cite{fasttrack},
the clock $\Ww_x$ is always an epoch.
The underlying assumption for such a simplification is that all 
the events that write to a given memory location are 
totally ordered with respect to $\hb{}$.
This assumption, however, need not hold beyond the first HB race.
After the first race is encountered, two $\wt(x)$ events 
$e$ and $e'$ may be unordered by both $\hb{}$ and $\shb{}$.
In Algorithm~\ref{algo:epoch}, $\Ww_x$ has an epoch representation
if and only if the last write event $e$ on $x$
is such that $e' \shb{} e$ for every event $e'$
of the form $e' = \ev{\cdot, \wt(x)}$ in the trace seen so far.
When performing a write event $e = \ev{t, \wt(x)}$,
if $\Ww_x$ satisfies $\Ww_x \cle \Cc_t$ (Line 15), 
then the event $e$ is ordered after
all previous $\wt(x)$ events, and thus,
in this case, $\Ww_x$ is converted to an epoch representation 
independent of its original representation (see Line 16).
Otherwise, there are at least two $\wt(x)$ events that are
not ordered by $\shb{}$ and thus $\Ww_x$ becomes a full-fledged vector clock
(Lines 20 and 22).

\item
The clock $\Rr_x$ is only \emph{semi-adaptive} --- we do not switch back to
epoch representation once the clock $\Rr_x$ takes up a vector-time value.
The clock $\Rr_x$ is initialized to be an epoch.
When processing a read event $e = \ev{t, \rd(x)}$, if the algorithm
determines that there is a read event $e' = \ev{t', \rd(x)}$ observed earlier
such that $e' \not\leq_\textsf{SHB} e$, then the clock $\Rr_x$
takes a vector-time representation.
After this point, $\Rr_x$ stays in the vector clock representation forever.
The $\Rr_x$ clock is an epoch only if 
all the reads of $x$ observed are ordered totally by $\shb{}$.
Thus, in order to determine if $\Rr_x$ can be converted back to an epoch representation,
one needs to check if $(\Rr_x \cle \Cc_t)$ every time a read event is processed.
Since this is an expensive additional comparison and because most traces
from real-world examples are dominated by read events, we avoid such a check
and force $\Rr_x$ to be only semi-adaptive. 
This is similar to the \fasttrack~algorithm.
\end{itemize}
% Complete algorithm implementing this optimization can be found in Appendix~\ref{app:epoch}.
% \input{pseudocode2col_epoch}

As with \textsc{FastTrack}, the epoch optimization for $\shb{}$
is sound and does not lead to any loss of precision --- the optimized
algorithm (Algorithm~\ref{algo:epoch}) declares a race at an event $e$ 
iff the corresponding unoptimized algorithm (Algorithm~\ref{algo:vc}) declares a race at $e$.

One must however note that the new clock $\LW_x$
does not have an adaptive representation, and is always
required to be a vector clock. 
One can think of $\LW_x$ to be similar to, say, the clocks $\Ll_\lk$.
These clocks are used to maintain the partial order,
unlike the clocks $\Rr_x$ or $\Ww_x$ which are only
used to check for races.
Thus, one needs the full generality of vector times for $\LW_x$.

\section{Experiments}
\label{sec:exp}
%!TEX root = main.tex

We first describe our implementation to detect
$\hb{}$-schedulable races.
We then present a brief description of the chosen benchmarks
and finally the results of evaluating our implementation
on these benchmarks.

\subsection{Implementation}
\label{subsec:impl}

We have implemented our SHB-based race detection algorithms
(Algorithm~\ref{algo:vc} and Algorithm~\ref{algo:epoch}) 
in our tool \tool, which is publicly available at~\cite{rapid}.
{\tool} is written in Java and supports
analysis on traces generated by the instrumentation and logging 
functionality provided by \rvpredict~\cite{rvpredict} to generate traces
from Java programs.
The traces generated by \rvpredict~contain read, write,
fork, join, acquire and release events.
We assume that the traces are sequentially consistent,
similar to the assumption made by~\cite{rv2014}. 
% \rvpredict, at least in theory, detects more races than 
% any existing sound race predictive algorithm.
We compare the performance of five dynamic race detection algorithms
to demonstrate the effectiveness of SHB-based sound reasoning:
% We have implemented our SHB-based race detection algorithms in the tool
% \tool~\cite{rapid}. {\tool} is written in Java 
% and uses the instrumentation and logging functionality
% of \rvpredict~\cite{rvpredict} for Java programs. We compare the
% performance of four race detection algorithms to demonstrate the
% effectiveness of SHB-based reasoning.
\begin{description}
\item[HB] We implemented the 
  \textsc{Djit}$^+$ algorithm for computing the $\hb{}$-partial order
  and detecting HB-races, in our tool~\tool.  As with popular
  implementations of \textsc{Djit}$^+$, our implementation
  of \textsc{Djit}$^+$ discovers all $\hb{}$-unordered pairs of
  conflicting events. This serves as a base line to demonstrate how
  many false positives would result, if one considered all HB-races
  (instead of $\hb{}$-schedulable races).  This algorithm is same as
  Algorithm~\ref{algo:vc} except that the lines involving the clock
  $\LW_x$ (Lines 5, 21 and 28) are absent.
\item[SHB] This is the implementation of Algorithm~\ref{algo:vc}
  in our tool~\tool. The soundness guarantee of Algorithm~\ref{algo:vc}
  (Theorem~\ref{thm:correct-races}) ensures that our implementation
  reports only (and all) $\hb{}$-schedulable races and thus reports no false alarms.
  % This algorithm also analyzes the complete trace (instead of terminating
  % after the first race) and does not raise any false alarms.
  % We have implemented both a vector clock-based
  % algorithm and its epoch-based optimization.
% \item[SHB(S)] This is the ``sound'' version of the SHB algorithm that checks 
%   conditions~\ref{lbl:sufficient-read} and~\ref{lbl:sufficient-write}
%   of Theorem~\ref{thm:SHBSoundness}. This algorithm \emph{only}
%   reports races that are schedulable, but it may miss some read-write
%   races that are schedulable. We have implemented both a vector clock
%   version and its epoch-based optimization for this algorithm.
\item[FHB] This is the algorithm that mimics a software developer's strategy 
  when using HB-race detectors. 
  This algorithm is a slight variant of the \textsc{Djit}$^+$
  algorithm and is implemented in our tool~\tool.
  Every time an HB-race is discovered,
  the algorithm force orders the events in race, before analyzing subsequent
  events in the trace.
  When processing a read event $e = \ev{t, \rd(x)}$, if the
  algorithm discovers a race (that is, if the check $\neg(\Ww_x \cle \Cc_t)$ passes),
  the algorithm reports a race and also updates the clock
  $\Cc_t$ as $\Cc_t := \Cc_t \sqcup \Ww_x$.
  Similarly, at a write event, the algorithm updates the clock
  $\Cc_t$ as $\Cc_t := \Cc_t \sqcup \Rr_x$ (resp. $\Cc_t := \Cc_t \sqcup \Ww_x$)
  if the check $\neg(\Rr_x \cle \Cc_t)$ (resp. $\neg(\Ww_x \cle \Cc_t)$) passes.
  % Recall that a vector clock algorithm (like $\textsc{Djit}^+$ or \fattrack)
  % detects a race $(e_1, e_2)$ when it observes the second event $e_2$
  % and checks 
  This algorithm is sound --- all races reported by this algorithm are
  schedulable, but it may fail to identify some races that are
  schedulable. The complete description of FHB is presented in Algorithm~\ref{algo:fhb}.
  \SetKwProg{myalg}{procedure}{}{}

\begin{algorithm*}
\begin{multicols}{2}
\SetKwFunction{facq}{acquire}
\SetKwFunction{frel}{release}
\SetKwFunction{fork}{fork}
\SetKwFunction{join}{join}
\SetKwFunction{read}{read}
\SetKwFunction{wr}{write}
\SetKwFunction{init}{Initialization}

\myalg{\init}{
    \lForEach{$t$}{
    $\Cc_t$ := $\bot[1/t]$
    } 
    \lForEach{$\ell$}{
    $\Ll_\ell$ := $\bot$ % $\lambda t: 0$
    } 
    \For{$x \in \vars{}$}{ 
    $\Rr_x$ := $\bot$; \\% $\lambda t: 0$ ; 
    $\Ww_x$ := $\bot$;  % $\lambda t: 0$ ;
    }
}

\myalg{\facq{$t$, $\ell$}}{ %
    $\Cc_t$ := $\Cc_t \mx \Ll_\ell$ ;
}

\myalg{\frel{$t$, $\ell$}}{
    $\Ll_\ell$ := $\Cc_t$ ; \\
    $\Cc_t(t)$ := $\Cc_t(t) + 1$ ; {\scriptsize \texttt{(* next event *)}}
}

\myalg{\fork{$t$, $u$}}{
    $\Cc_u$ := $\Cc_t[1/u]$ ; \\
    $\Cc_t(t)$ := $\Cc_t(t) + 1$ ; {\scriptsize \texttt{(* next event *)}}
}

\myalg{\join{$t$, $u$}}{
    $\Cc_t$ := $\Cc_t \mx \Cc_u$ ;
}

\myalg{\read{$t$, $x$}}{
    \If{$\neg (\Ww_x \cle \Cc_t)$}{
      declare `\textbf{race with write}';\\
      $\Cc_t$ := $\Cc_t \sqcup \Ww_x$; {\scriptsize \texttt{(* force ordering *)}}\\
    }
    $\Cc_t$ := $\Cc_t \mx \LW_x$; \\
    $\Rr_x(t)$ := $\Cc_t(t)$; \\
}

\myalg{\wr{$t$, $x$}}{
    \If{$\neg (\Rr_x \cle \Cc_t)$}{
      declare `\textbf{race with read}'; \\
      $\Cc_t$ := $\Cc_t \sqcup \Rr_x$; {\scriptsize \texttt{(* force ordering *)}} \\
    }
    \If{$\neg (\Ww_x \cle \Cc_t)$}{
      declare `\textbf{race with write}'; \\
      $\Cc_t$ := $\Cc_t \sqcup \Ww_x$; {\scriptsize \texttt{(* force ordering *)}} \\
    } 
    $\Ww_x(t)$ := $\Cc_t(t)$; \\
}
\end{multicols}
\vspace*{0.25cm}
\caption{\textit{Vector Clock for FHB race detection}}
\label{algo:fhb}
\end{algorithm*}
\item[WCP] WCP or Weak Causal Precedence~\cite{wcp2017} is another sound
  partial order that can be employed for predictive data race
  detection.  WCP is weaker than both, its precursor CP~\cite{cp2012},
  and HB.  That is, whenever HB or CP detect the presence of a race in
  a trace, WCP will also do so, and in addition, there are traces when
  WCP can correctly detect the presence of a race when neither HB or
  CP can.  Nevertheless, WCP (and CP) also suffer from the same
  drawback as HB --- the soundness guarantee applies only to the first
  race. As a result, races beyond the first one, detected by WCP (or
  CP) may not be real races.  WCP admits a linear time vector clock
  algorithm and is also implemented in~\tool.  
\item[RVPredict] \rvpredict's race detection technology
relies on maximal causal models~\cite{rv2014}. 
\rvpredict~is sound and does not report any false alarms.
Besides, \rvpredict, at least in theory,
guarantees to detect more races than any other sound race prediction tool,
and thus more races than Algorithm~\ref{algo:vc} theoretically.
\rvpredict~encodes the problem of race detection as a logical formula
and uses an SMT solver to check for races.
\rvpredict~can analyze the traces generated using its logging functionality,
and thus is a natural choice for comparison.
\end{description}

Besides the vector clock algorithms (HB, SHB, FHB, WCP) described
above, we also implemented the epoch optimizations for HB and SHB
in~\tool.

\subsection{Benchmarks}

We measure the performance of our algorithms against traces drawn from
a wide variety of benchmark programs (Column 1 in
Table~\ref{tab:metadata}) that have previously been used to measure
the performance of other race detection
tools~\cite{cp2012,rv2014,wcp2017}. 
The set of benchmarks have been derived from different suites.  
The examples \textsf{airlinetickets} to \textsf{pingpong} are
small-sized, and belong to the IBM Contest benchmark
suite~\cite{Farchi2003}, with lines of code roughly varying from 40 to
0.5M.  The benchmarks \textsf{moldyn}
and \textsf{raytracer} are drawn from the Java
Grande Forum benchmark suite~\cite{JGF2001} and are medium-sized with about 3K lines of code.  
The third set of benchmarks correspond to real-world software 
applications and include Apache FTPServer, W3C Jigsaw web server, 
Apache Derby, and others (\textsf{xalan} to \textsf{eclipse}) 
derived from the DaCaPo benchmark suite (version 9.12)~\cite{DaCapo2006}.

In Table~\ref{tab:metadata}, we also describe the characteristics of
the generated traces that we use for analyzing our algorithms.  The
number of threads range from 3-12, the number of lock objects can be
as high as 8K.  The distinct memory locations accessed (Column 5) in
the traces can go as high as 10M.  The traces generated are dominated
by access events, with the majority of events being read events
(compare Columns 6, 7 and 8).

\subsection{Setup}
 
Our experiments were conducted on an 8-core 2.6GHz 46-bit Intel
Xeon(R) Linux machine, with HotSpot 1.8.0 64-Bit Server as the JVM and
50 GB heap space.  Using \rvpredict's logging functionality, we
generated one trace per benchmark and analyzed it with the various race
detection engines: HB, SHB, FHB, WCP and \rvpredict.

Our evaluation is broadly designed to evaluate our approach
based on the following aspects:

\begin{enumerate}
  \item \textbf{Reducing false positives:} Dynamic race detection
  tools based on Eraser style lockset based
  analysis~\cite{savage1997eraser} are known to scale better than
  those based on happens-before despite careful optimizations like the
  use of epochs~\cite{fasttrack}.  One of the main reasons for the
  popularity of HB-based race detection tools such
  as \fasttrack~\cite{fasttrack} and
  ThreadSanitizer~\cite{threadsanitizer} is the ability to produce
  reliable results (no false positives).  However, as pointed out in
  Section~\ref{sec:intro}, HB based analysis can report false races
  beyond the first race discovered.  The purpose of detecting
  $\hb{}$-schedulable races, instead of all HB-races, is to ensure
  that only correct races are reported.  However, since our algorithm
  for detecting $\hb{}$-schedulable races tracks additional vector
  clocks (namely $\LW_x$ for every memory location $x$), we would like
  to demonstrate the importance of such an additional book-keeping for
  ensuring soundness of happens-before based reasoning.

  \item \textbf{Prediction power:} As described in
  Section~\ref{sec:intro}, a na\"ive fix to the standard HB race
  detection algorithm is to employ the FHB algorithm --- after a race
  is discovered at an event, order the event with all conflicting
  events observed before it.  We would like to examine if the use of
  $\shb{}$-based reasoning enhances prediction power by detecting more
  races than this na\"ive strategy.  Further, we would like to
  evaluate if more powerful approaches like the use of SMT solvers
  in \rvpredict~give significantly more benefit as compared to our
  linear time streaming algorithm.

  \item \textbf{Scalability:} While Algorithm~\ref{algo:vc} runs in
  linear time, it tracks additional clocks ($\LW_x$ for every memory
  location $x$ accessed in the trace) over the standard HB vector
  clock algorithm.  Since this can potentially slow down analysis, we
  would like to evaluate the performance overhead due to this
  additional book-keeping.

  \item \textbf{Epoch optimization:} The standard epoch optimization
  popularized by \fasttrack~\cite{fasttrack} is designed to work for
  the case when all the writes to a memory location are totally
  ordered.  While this is true until the first race is discovered,
  this condition may not be guaranteed after the first race.  We will
  evaluate the effectiveness of our adaptation of this optimization to
  work beyond the first race.
\end{enumerate}

% \ucomment{Umang continues from this point.}

\subsection{Evaluation}

%!TEX root = main.tex

\begin{table*}[t!]
\captionsetup{font=small}

\scalebox{0.85}{
\centering

\begin{tabular}{!{\VRule[1pt]}c|c!{\VRule[1pt]}c|c|c!{\VRule[1pt]}c|c|c|c|c|c!{\VRule[1pt]}}

\specialrule{1pt}{0pt}{0pt}
1 & 2 & 3 & 4 & 5 & 6 & 7 & 8 & 9 & 10 & 11\\ 
\specialrule{1pt}{0pt}{0pt}

\cellcolor[HTML]{DDDDDD} Program
& \cellcolor[HTML]{DDDDDD} \, LOC \,
& \cellcolor[HTML]{DDDDDD} Thrds
& \cellcolor[HTML]{DDDDDD} Locks
& \cellcolor[HTML]{DDDDDD} Vars
& \multicolumn{6}{c!{\VRule[1pt]}}{\cellcolor[HTML]{DDDDDD}{Events}} \\

\cmidrule[0.5pt]{6-11}

\rowcolor[HTML]{DDDDDD} 
 \cellcolor[HTML]{DDDDDD} 
& \cellcolor[HTML]{DDDDDD}
& \cellcolor[HTML]{DDDDDD}
& \cellcolor[HTML]{DDDDDD} 
& \cellcolor[HTML]{DDDDDD} 
& \cellcolor[HTML]{DDDDDD} Total
& \cellcolor[HTML]{DDDDDD} Read
& \cellcolor[HTML]{DDDDDD} Write
& \cellcolor[HTML]{DDDDDD} Synch.
& \cellcolor[HTML]{DDDDDD} \, Fork \,
& \cellcolor[HTML]{DDDDDD} \, Join \, \\

\specialrule{1pt}{0pt}{0pt}

\textsf{airlinetickets} & 83 & 4 & 0 & 44 & 137 & 77 & 48 & 0 & 12 & 0 \\ 

\textsf{array} & 36 & 3 & 2 & 30 & 47 & 8 & 30 & 3 & 3 & 0 \\ 

\textsf{bufwriter} & 199 & 6 & 1 & 471 & 22.2K & 15.8K & 3.6K & 1.4K & 6 & 4 \\ 

\textsf{bubblesort} & 274 & 12 & 2 & 167 & 4.6K & 4.0K & 404 & 121 & 27 & 0 \\ 

\textsf{critical} & 63 & 4 & 0 & 30 & 55 & 18 & 31 & 0 & 4 & 2 \\ 

\textsf{mergesort} & 298 & 5 & 3 & 621 & 3.0K & 2.0K & 914 & 55 & 5 & 3 \\ 

\textsf{pingpong} & 124 & 6 & 0 & 51 & 147 & 57 & 71 & 0 & 19 & 0 \\ 

\textsf{moldyn} & 2.9K & 3 & 2 & 1.2K & 200.0K & 182.8K & 17.2K & 31 & 3 & 1 \\ 

\textsf{raytracer} & 2.9K & 3 & 8 & 3.9K & 15.8K & 10.4K & 5.3K & 60 & 3 & 1 \\ 

\textsf{derby} & 302K & 4 & 1112 & 185.6K & 1.3M & 879.5K & 404.5K & 31.2K & 4 & 2 \\ 

\textsf{ftpserver} & 32K & 11 & 301 & 5.5K & 49.0K & 30.0K & 7.8K & 5.6K & 11 & 4 \\ 

\textsf{jigsaw} & 101K & 11 & 275 & 103.5K & 3.1M & 2.6M & 413.5K & 5.9K & 13 & 4 \\ 

\textsf{xalan} & 180K & 6 & 2491 & 4.4M & 122.0M & 101.7M & 18.3M & 1M & 7 & 5 \\ 

\textsf{lusearch} & 410K & 7 & 118 & 5.2M & 216.4M & 162.1M & 53.9M & 206.6K & 7 & 0 \\ 

\textsf{eclipse} & 560K & 14 & 8263 & 10.6M & 87.1M & 72.6M & 12.9M & 765.4K & 16 & 3 \\ 

\specialrule{1pt}{0pt}{0pt}

% Total &  &  &  &  & 430.3M & 340.2M & 85.9M & 2.1M & 140 & 29 \\ 
\multicolumn{5}{!{\VRule[1pt]}c!{\VRule[1pt]}}{Total}  & 430.3M & 340.2M & 85.9M & 2.1M & 140 & 29 \\ 
\specialrule{1pt}{0pt}{0pt}
\end{tabular}
}
\caption{
Benchmarks and metadata of the traces generated.
Columns 1 and 2 describe the name and the lines of code in the source code of the chosen benchmarks.
Column 3, 4 and 5 describe respectively the number of \emph{active}\protect\footnotemark
threads, locks, and memory locations in the traces generated by the corresponding program in Column 1.
Column 6 reports the total number of events in the trace.
Columns 7, 8, 9, 10 and 11 respectively denote
the number of read, write, acquire (or release), fork and join events.
}
\label{tab:metadata}
\end{table*}

\footnotetext{a thread is active if there is an event $e = \ev{t, op}$ performed by the thread $t$ in the trace generated}
%!TEX root = main.tex

\begin{table*}[t]
\captionsetup{font=small}

\scalebox{0.83}{
\centering

\begin{tabular}{!{\VRule[1pt]}c|c!{\VRule[1pt]}!{\VRule[1pt]}c|c|c|c|c|c!{\VRule[1pt]}!{\VRule[1pt]}c|c|c|c!{\VRule[1pt]}}

\specialrule{1pt}{0pt}{0pt}
1 & 2 & 3 & 4 & 5 & 6 & 7 & 8 & 9 & 10 & 11 & 12\\ 
\specialrule{1pt}{0pt}{0pt}
 
\cellcolor[HTML]{DDDDDD} 
& \cellcolor[HTML]{DDDDDD} 
& \multicolumn{6}{c!{\VRule[1pt]}!{\VRule[1pt]}}{\cellcolor[HTML]{DDDDDD}{Races}} 
& \multicolumn{4}{c!{\VRule[1pt]}}{\cellcolor[HTML]{DDDDDD}{Warnings}}   \\

\cmidrule[0.5pt]{3-12}

 \cellcolor[HTML]{DDDDDD} Program 
& \cellcolor[HTML]{DDDDDD} \#Events
& \cellcolor[HTML]{DDDDDD}  \;\! HB \;\!
& \cellcolor[HTML]{DDDDDD}  SHB
& \cellcolor[HTML]{DDDDDD}  FHB
& \cellcolor[HTML]{DDDDDD}  WCP
& \multicolumn{2}{c!{\VRule[1pt]}!{\VRule[1pt]}}{\cellcolor[HTML]{DDDDDD}{\rvpredict}} 
& \cellcolor[HTML]{DDDDDD} HB
& \cellcolor[HTML]{DDDDDD} SHB
& \cellcolor[HTML]{DDDDDD} FHB
& \cellcolor[HTML]{DDDDDD} WCP \\

\cmidrule[0.5pt]{7-8}

\cellcolor[HTML]{DDDDDD} 
& \cellcolor[HTML]{DDDDDD} 
& \cellcolor[HTML]{DDDDDD} 
& \cellcolor[HTML]{DDDDDD} 
& \cellcolor[HTML]{DDDDDD} 
& \cellcolor[HTML]{DDDDDD} 
& \cellcolor[HTML]{DDDDDD} 1K/60s %{\small 1K/60s}
& \cellcolor[HTML]{DDDDDD} 10K/240s %{\small 10K/240s}
& \cellcolor[HTML]{DDDDDD} 
& \cellcolor[HTML]{DDDDDD} 
& \cellcolor[HTML]{DDDDDD} 
& \cellcolor[HTML]{DDDDDD} \\

\specialrule{1pt}{0pt}{0pt}

\textsf{airlinetickets} & 137 & 6 & 6 & 3 & 6 & 6 & 6 & 8 & 8 & 5 & 8 \\ 

\textsf{array} & 47 & 0 & 0 & 0 & 0 & 0 & 0 & 0 & 0 & 0 & 0\\ 

\textsf{bufwriter} & 22.2K & 2 & 2 & 2 & 2 & 2 & 0 & 8 & 8 & 8 & 8 \\ 

\textsf{bubblesort} & 4.6K & 6 & 6 & 6 & 6 & 6 & 0 & 602 & 269 & 100 & 612\\ 

\textsf{critical} & 55 & 8 & 8 & 1 & 8 & 8 & 8 & 3 & 3 & 1 & 3\\ 

\textsf{mergesort} & 3.0K & 3 & 1 & 1 & 3 & 1 & 2 & 52 & 1 & 1 & 52 \\ 

\textsf{pingpong} & 147 & 3 & 3 & 3 & 3 & 3 & 3 & 11 & 8 & 8 & 11\\ 

\textsf{moldyn} & 200.0K & 44 & 2 & 2 & 44 & 2 & 2 & 24657 & 103 & 103 & 24657 \\ 

\textsf{raytracer} & 15.8K & 3 & 3 & 3 & 3 & 2 & 3 & 118 & 8 & 8 & 118\\ 

\textsf{derby} & 1.3M & 26 & 13 & 11 & 26 & 12 & - & 89 & 29 & 28 & 89\\ 

\textsf{ftpserver} & 49.0K & 35 & 23 & 22 & 35 & 10 & 12 & 143 & 69 & 69 & 144\\ 

\textsf{jigsaw} & 3.1M & 8 & 4 & 4 & 10 & 4 & 2 & 14 & 4 & 4 & 17\\ 

\textsf{xalan} & 122.0M & 16 & 12 & 10 & 18 & 8 & 8 & 86 & 31 & 21 & 98\\ 

\textsf{lusearch} & 216.4M & 160 & 52 & 28 & 160 & 0 & 0 & 751002 & 232 & 119 & 751002\\ 

\textsf{eclipse} & 87.1M & 64 & 61 & 31 & 66 & 5 & 0 & 173 & 164 & 103 & 201\\ 

\specialrule{1pt}{0pt}{0pt}

Total & 430.3M & 384 & 196 & 127 & 390 & 69 & 46 & 776966 & 937 & 578 & 777020\\ 
\specialrule{1pt}{0pt}{0pt}

\end{tabular}
}
\caption{
Number of races detected and warnings raised.
% Column 1 and 2 describe the benchmarks (name and lines of code respectively). 
Column 1 and 2 denote the benchmarks and the size of the traces generated.
Columns 3, 4, 5 and 6 respectively report the number of distinct program 
location pairs for which there are pair of events in
a race, as identified by HB, SHB, FHB and WCP.
Columns 7 and 8 denote the races reported by \rvpredict~when run with 
the parameters (\texttt{window-size=1K}, \texttt{solver-timeout=60s})
and (\texttt{window-size=10K}, \texttt{solver-timeout=240s}).
Columns 9, 10, 11 and 12 respectively denote the number of warnings 
generated when running the vector clock algorithms for
detecting races using HB (unsound), SHB (sound and complete for $\hb{}$-schedulable races), 
FHB (naive algorithm that forces an order after every race discovered) and WCP analyses.
}
\label{tab:races}
\end{table*}

Our experimental results are summarized in
Table~\ref{tab:metadata}, Table~\ref{tab:races} and Table~\ref{tab:times}. 
Table~\ref{tab:metadata} describes information about
generated execution logs.
Table~\ref{tab:races} depicts the number of races and warnings
raised by the different race detection algorithms.
Columns 3-8 in Table~\ref{tab:races} report the number of distinct pairs
$(pc_1, pc_2)$ of program locations corresponding
to an identified data race.
That is, for every \emph{event} race pair $(e_1, e_2)$ identified by the different
race detection algorithms, we identify the pair of \emph{program locations}
that give rise to this event pair and report the total
number of such program location pairs (counting the
pairs $(pc_1, pc_2)$ and $(pc_2, pc_1)$ only once).
Since each of the vector clock algorithms (HB, SHB, FHB and WCP)
only report whether the event being processed is in race
with some earlier event, we need to perform a separate analysis step
using the vector timestamps, to determine the actual pair of events
(and thus the corresponding pair of program locations) in race.
In Columns 9, 10, 11 and 12 in Table~\ref{tab:races} we report the number of \emph{warnings}
raised by the four vector clock algorithms---HB, SHB, FHB and WCP respectively.
A warning is raised when at a read/write event $e$, we determine if
the event $e$ is in race with an earlier event, 
counting multiple warnings for a single event only once.
% As can be noted, naively running HB vector clock algorithm, can lead to
% a lot of warnings (see Column 2, for example, in \textsf{moldyn} or \textsf{lusearch}),
% and most of these are potentially spurious.
In Table~\ref{tab:times}, Columns 2, 5, 8, 9, 10 and 11 respectively
report the time taken by
different analyses engines --- HB, SHB, FHB, WCP and \rvpredict--- on the trace generated.
We also measure the time taken by the epoch optimizations
for both HB and SHB vector clock algorithms (Columns 4 and 7 respectively)
and report the speedup thus obtained over the na\"ive vector clock algorithms 
(Columns 4 and 7 respectively). 
When analyzing the generated traces using WCP, we filter out events
that are thread local; this does not affect any races.
The memory requirement of a na\"ive vector clock algorithm for WCP, 
as described in~\cite{wcp2017} can be a bottleneck and
removing thread local events allowed us to analyze the larger traces 
(\textsf{xalan}, \textsf{lusearch} and \textsf{eclipse}) without any memory blowup.
We next discuss our results in detail.

\subsubsection{Reducing false positives}
First, observe that both the number of races reported (Columns 3, 4
and 5 in Table~\ref{tab:races}) and the number of warnings raised
(Columns 8, 9 and 10) by HB, SHB and FHB are monotonically decreasing,
as expected --- HB detects all $\hb{}$-schedulable races but
additional false races, SHB detects exactly the set of
$\hb{}$-schedulable races and FHB detects a subset of
$\hb{}$-schedulable races.  Next, the number of races reported by HB
can be way higher than the actual number of $\hb{}$-schedulable races
(see \textsf{moldyn} and \textsf{lusearch}).  Similarly, the number of
warnings raised can be an order of magnitude larger than those raised
by either SHB or FHB.  Clearly, many of these warnings are potentially
spurious.  Thus, an incorrect use of the popular HB algorithm can
severely hamper developer productivity, and completely defies the
point of using a sound race detection analysis technique.
Further, in each of the benchmarks, both the set of races as well 
as the set of warnings
reported by WCP were a superset of those reported by HB. 
This follows from the fact that WCP is a strictly weaker relation than HB.

%!TEX root = main.tex

\begin{figure}[b]
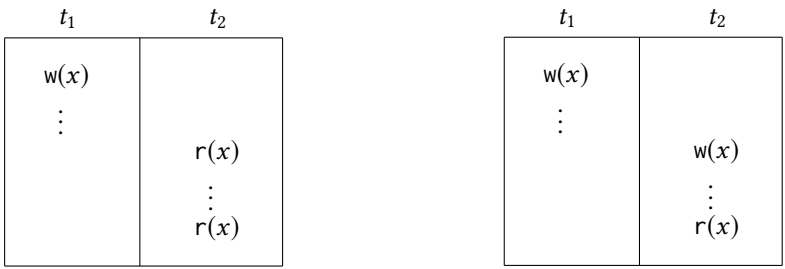

\centering
\begin{subfigure}{.4\textwidth}
  \centering
  \vspace{0.4cm}
    \executionnonumber{2}{
      \figevnonumber{1}{\wt(x)}
      \figevnonumber{1}{\,\,\,\vdots}
      \figevnonumber{2}{\rd(x)}
      \figevnonumber{2}{\,\,\,\vdots}
      \figevnonumber{2}{\rd(x)}
  }
  \caption{Incorrect race reported by HB but not by SHB}
  \label{fig:pattern1}
\end{subfigure}%
\hspace{1cm}
\begin{subfigure}{.4\textwidth}
  \centering
  \vspace{0.15in}
  \executionnonumber{2}{
      \figevnonumber{1}{\wt(x)}
      \figevnonumber{1}{\,\,\,\vdots}
      \figevnonumber{2}{\wt(x)}
      \figevnonumber{2}{\,\,\,\vdots}
      \figevnonumber{2}{\rd(x)}
  }
  \caption{Correct race missed by FHB but detected by SHB}
  \label{fig:pattern2}
\end{subfigure}
\caption{Common race patterns found in the benchmarks}
\label{fig:patterns}
% \vspace*{.5in}
\end{figure}

While Theorem~\ref{thm:SHBSoundness} guarantees that the each of the additional
race pairs reported by HB (over those reported by SHB) 
cannot be scheduled in any correct reordering of the observed trace 
that respects the induced $\hb{}$ partial order, it does not guarantee
that these extra races cannot be scheduled in \emph{any} correct reordering.
In order to see if the extra races reported by HB (Column 3 in Table~\ref{tab:races})
can be scheduled in a correct reordering that does not respect $\hb{}$ order,
we manually inspected the traces (annotated with their
vector timestamps) of \textsf{mergesort}, \textsf{moldyn}, \textsf{derby}, 
\textsf{ftpserver}, and \textsf{jigsaw}. In each of these benchmarks,
we found that all the extra race pairs reported by HB can indeed
\emph{not} be scheduled in \emph{any} correct reordering
(whether or not the correct reordering respects the induced $\hb{}$ partial order).
A common pattern that helped us conclude this observation has been depicted in
Fig.~\ref{fig:pattern1}.  Here, the trace writes to a memory location
$x$ in a thread $t_1$ (event $e_1$).  Then, sometime later, another
event $e_2$ performed by a different thread $t_2$ reads the value
written by $e_1$. This is then followed by other events of thread
$t_2$, not pertaining to memory location $x$.  Finally, thread $t_2$
reads the memory location $x$ again in event $e_3$.  This pattern is
commonly observed when a thread reads a shared variable (here, this
corresponds to the event $e_2$), takes a branch depending upon the
value observed and then accesses the shared memory again within the
branch.  HB misses this dependency relation thus induced, and
incorrectly reports that the pair $(e_1, e_3)$ is in race.  SHB, on
the other hand, correctly orders $e_1 \shb{} \ltho{}(e_3)$, and does not
report a race.

The two extra races reported by WCP but not by HB in \textsf{jigsaw}
could not be confirmed to be false positives.
Further, we did not inspect the extra races reported by HB or WCP (over SHB) 
in \textsf{xalan}, \textsf{lusearch}
and \textsf{eclipse} owing to time constraints.

%!TEX root = main.tex

\begin{table*}[t]
\captionsetup{font=small}

\scalebox{0.83}{
\centering

\begin{tabular}{!{\VRule[1pt]}c!{\VRule[1pt]}!{\VRule[1pt]}c|c|c!{\VRule[1pt]}c|c|c!{\VRule[1pt]}c!{\VRule[1pt]}c!{\VRule[1pt]}c|c!{\VRule[1pt]}}

\specialrule{1pt}{0pt}{0pt}
1 & 2 & 3 & 4 & 5 & 6 & 7 & 8 & 9 & 10 & 11\\ 
\specialrule{1pt}{0pt}{0pt}

 \cellcolor[HTML]{DDDDDD} 
& \multicolumn{3}{c!{\VRule[1pt]}}{\cellcolor[HTML]{DDDDDD}{HB}} 
& \multicolumn{3}{c!{\VRule[1pt]}}{\cellcolor[HTML]{DDDDDD}{SHB}} 
& \cellcolor[HTML]{DDDDDD} \, FHB \,
& \cellcolor[HTML]{DDDDDD} \, WCP \,
& \multicolumn{2}{c!{\VRule[1pt]}}{\cellcolor[HTML]{DDDDDD}{\rvpredict}} \\

\cmidrule[0.5pt]{2-7}

\cmidrule[0.5pt]{10-11}

\cellcolor[HTML]{DDDDDD} Program 
& \cellcolor[HTML]{DDDDDD} VC
& \cellcolor[HTML]{DDDDDD} Epoch
& \cellcolor[HTML]{DDDDDD} {\small Speed-up}
& \cellcolor[HTML]{DDDDDD} VC
& \cellcolor[HTML]{DDDDDD} Epoch
& \cellcolor[HTML]{DDDDDD} {\small Speed-up}
& \cellcolor[HTML]{DDDDDD} 
& \cellcolor[HTML]{DDDDDD} 
& \cellcolor[HTML]{DDDDDD} \;{\small 1K/60s}\;
& \cellcolor[HTML]{DDDDDD} {\small 10K/240s}\\

\cellcolor[HTML]{DDDDDD} 
& \cellcolor[HTML]{DDDDDD} (s)
& \cellcolor[HTML]{DDDDDD} (s)
& \cellcolor[HTML]{DDDDDD} 
& \cellcolor[HTML]{DDDDDD} (s)
& \cellcolor[HTML]{DDDDDD} (s)
& \cellcolor[HTML]{DDDDDD} 
& \cellcolor[HTML]{DDDDDD} (s)
& \cellcolor[HTML]{DDDDDD} (s)
& \cellcolor[HTML]{DDDDDD} (s)
& \cellcolor[HTML]{DDDDDD} (s)\\

\specialrule{1pt}{0pt}{0pt}

\textsf{airlinetickets} & 1.11 & 1.4 & - & 2.08 & 0.25 & 8.32x & 2.33 & 0.35 & 1.34 & 1.33 \\ 

\textsf{array} & 1.13 & 1.66 & - & 0.55 & 0.25 & 2.2x & 1.98 & 0.77 & 1.28 & 1.26 \\ 

\textsf{bufwriter} & 0.54 & 1.63 & - & 0.56 & 0.46 & 1.22x & 0.96 & 1.3 & 3.61 & 1879 \\ 

\textsf{bubblesort} & 0.47 & 0.29 & 1.62x & 0.78 & 0.34 & 2.3x & 0.51 & 0.63 & 2.46 & 652 \\ 

\textsf{critical} & 0.68 & 0.24 & 2.8x & 0.28 & 0.56 & - & 0.54 & 0.94 & 2.14 & 0.52 \\ 

\textsf{mergesort} & 0.4 & 1.58 & - & 0.43 & 0.35 & 1.23x & 0.7 & 0.8 & 0.67 & 0.81 \\ 

\textsf{pingpong} & 0.31 & 0.31 & 1x & 2.15 & 0.33 & 6.5x & 0.41 & 0.99 & 2.3 & 0.57 \\ 

\textsf{moldyn} & 1.45 & 2.78 & - & 1.04 & 1.49 & - & 1.77 & 1.81 & 2.27 & 2.94 \\ 

\textsf{raytracer} & 0.95 & 2.25 & - & 1.82 & 0.41 & 4.44x & 2.33 & 0.77 & 0.61 & 8.61 \\ 

\textsf{derby} & 4.27 & 2.91 & 1.46x & 4.39 & 3.05 & 1.44x & 5.56 & 9.24 & 72 & - \\ 

\textsf{ftpserver} & 0.84 & 0.84 & 1x & 0.87 & 1.33 & - & 2.94 & 1.64 & 1.23 & 164 \\ 

\textsf{jigsaw} & 23 & 8.1 & 2.84x & 10.04 & 7.35 & 1.37x & 9.31 & 11.11 & 1.66 & 245 \\ 

\textsf{xalan} & 217 & 152 & 1.42x & 222 & 184 & 1.21x & 291 & 290 & 44 & 420 \\ 

\textsf{lusearch} & 362 & 383 & - & 444 & 337 & 1.32x & 325 & 341 & 48 & 47 \\ 

\textsf{eclipse} & 525 & 200 & 2.63x & 238 & 188 & 1.27x & 168 & 512 & 25 & 951 \\ 

% \specialrule{1pt}{0pt}{0pt}

% Total & 430.3M & 1139.15 & 758.99 &  & 928.99 & 725.17 & & 813.34 & 208.57 & 4374.04 \\ 

\specialrule{1pt}{0pt}{0pt}
			
\end{tabular}
}
\caption{
Time taken by different race detection algorithms on traces 
generated by the corresponding programs in Column 1.
Column 2 denotes the taken for analyzing the entire trace
with the \textsc{Djit}$^+$ vector-clock algorithm.
Column 3 denotes the time taken by \fasttrack-style
optimization over the basic \textsc{Djit}$^+$
and Column 4 denotes the speedup thus obtained.
Column 5 denotes the times for the vector clock implementation of
Algorithm~\ref{algo:vc}. Column 6 and 7 denote the time
and speedup due to the epoch optimization for SHB (Algorithm~\ref{algo:epoch}).
A `-' in Columns 5 and 8 denote a downgraded performance due
to epoch optimization.
Column 8 denotes the time to analyze the traces using FHB
(forcing an order in HB) analysis.
Column 9 reports the time to analyze the traces using WCP partial order.
The analysis in Column 9 is performed by filtering out thread local events
and includes the time for this filtering.
Column 10 and 11 respectively denote the time taken by
\rvpredict~using the parameters (\texttt{window-size=1K}, \texttt{solver-timeout=60s})
and (\texttt{window-size=10K}, \texttt{solver-timeout=240s}).
A `-' in Column 11 denotes that \rvpredict~did not finish
within the set time limit of $4$ hours.
}
\label{tab:times}
\end{table*}

% airlinetickets,1.76
% array,0.31
% boundedbuffer,-
% bufwriter,0.92
% bubblesort,0.58
% critical,0.35
% mergesort,1.15
% pingpong,0.45
% moldyn,1.4
% raytracer,0.62
% derby,8.37
% ftpserver,1.53
% jigsaw,10.23
% xalan,325.41
% lusearch,149.07
% eclipse,166.87

% airlinetickets,0.35
% array,0.77
% boundedbuffer,-
% bufwriter,1.3
% bubblesort,0.63
% critical,0.94
% mergesort,0.8
% pingpong,0.99
% moldyn,1.81
% raytracer,0.77
% derby,9.24
% ftpserver,1.64
% jigsaw,11.11
% xalan,289.96
% lusearch,340.39
% eclipse,512.27

\subsubsection{Prediction Power}
The na\"{i}ve algorithm FHB, while sound, can miss a lot of real races
(Column 5 in Table~\ref{tab:races}) and has a poor prediction power as
compared to the sound SHB algorithm.  See for
example, \textsf{lusearch} and \textsf{eclipse} where FHB misses
almost half the races reported by SHB.  Next, observe that
while \rvpredict, in theory, is \emph{maximally sound}, it can miss a
lot of races, sometimes even more than the naive FHB strategy (Columns
6 and 7 in Table~\ref{tab:races}).  This is because \rvpredict~relies
on SAT solving to determine data races.  As a result, in order to
scale to large traces obtained from real world software,
\rvpredict~resorts to \emph{windowing} --- dividing the trace into
smaller chunks and restricting its analysis to these smaller chunks.
This strategy, while useful for scalability, can miss data races that
are spread far across in the trace, yet can be captured using
happens-before like analysis.  Besides, since the underlying
DPLL-based SAT solvers may not terminate within reasonable
time, \rvpredict~sets a timeout for the solver --- this means that
even within a given window, \rvpredict~can miss races if the SAT
solver does not return an answer within the set timeout.  All these
observations clearly indicate the power of $\shb{}$-based reasoning.

Again, based on our manual inspection of program traces, we depict a common
pattern found in Fig.~\ref{fig:pattern2}.
Here, first, a thread $t_1$ writes to a shared variable $x$ (event $e_1$).
This is followed by another write to $x$ in a different thread $t_2$
(event $e_2$). Finally, the next access to $x$ is a read event $e_3$
performed by thread $t_2$.
While FHB correctly reports the first write-write race $(e_1, e_2)$,
it fails to detect the write-read race $(e_1, e_3)$ because of the
artificial order imposed between $e_1$ and $e_2$.
SHB, on the other hand, reports both $(e_1, e_2)$
and $(e_1, e_3)$ as $\hb{}$-schedulable races.

% \ucomment{Come back. Describe windowing here}

% One particular striking observation that we would like to bring to attention
% is that naively running HB vector clock algorithm, can lead to
% a lot of warnings (see Column 4, for example, in \textsf{moldyn} or \textsf{lusearch}),
% and most of these are potentially spurious.
% This can severely hamper developer productivity, and completely defies the point
% of using a sound technique.
% At the same time, naively fixing this issue can lead to loss in prediction power.
% See, for example, Columns 8 and 9 of \textsf{eclipse}, where
% SHB identifies 88 distinct pairs, whereas FHB only identifies about half of these races.

% \input{table2}
% We now address the results in Table~\ref{tab:time}. 
\subsubsection{Scalability}
% As can be seen, the number of read and write events (Columns 5 and 6) can get really large.
First, the size of the traces, that SHB and the other three
linear time vector clock algorithms can handle, can be really large,
of the order of hundreds of millions (\textsf{xalan}, \textsf{lusearch}, etc.,).
In contrast, \rvpredict~fails to scale for large traces, even after
employing a windowing strategy.
This is especially pronounced for the larger traces 
(\textsf{bufwriter}, \textsf{derby}-\textsf{xalan}).
This suggests the power of using a linear time vector clock algorithm
for dynamic race detection for real-world applications.
The small and medium sized examples almost always
finish with a few seconds for each of HB, SHB, FHB and WCP.
The larger examples \textsf{xalan}, \textsf{lusearch}
and \textsf{eclipse} can take as much as 4-10 minutes
for the vector clock algorithms and 6-15 minutes for
\rvpredict's analysis with a window size of 10K and a solver timeout
of 4 minutes.

\subsubsection{Epoch optimization}
The epoch optimization is indeed effective in improving the
performance of vector-clock algorithms even when all the write events
to a memory location may not be totally ordered.  The speedups vary
from 1.2x to 8.3x on small and medium sized benchmarks and from 1.2x
to 2.9x on larger traces.  The speedup obtained for HB race detection
is, in general, less than in the case of SHB algorithm.  This is
expected since SHB is strictly stronger than HB --- every pair of
events ordered by $\hb{}$ is also ordered by $\shb{}$, and not every
pair of events ordered by $\shb{}$ may be ordered by $\hb{}$.  As a
result, a write event can get ordered after all previous write events
more frequently when using $\shb{}$ than when using $\hb{}$.  This
means that in the epoch optimization for SHB, the $\Ww_x$ clocks take
up epoch representation more frequently than in the HB algorithm, and
this difference is reflected in Columns 5 and 8 of
Table~\ref{tab:times}.

% Here, the performance can degrade on small traces (because of possibility of
% extra computation in maintaining epochs, and switching back an forth too frequently).
% However, as expected, on the larger traces, the performance benefit is significant
% with about 20\% speedup on most examples, with the exception of \textsf{xalan} where
% the epoch optimization leads to a slight slowdown.

\section{Related Work}
\label{sec:related}
%!TEX root = main.tex

The notion of correct reorderings to characterize causality in
executions has been derived from~\cite{cp2012,wcp2017}.
In~\cite{rv2014} a similar notion, called \emph{feasible traces}
encompasses a more general causality model based on control flow information.
Weak happens-before~\cite{sen2005detecting},  
Mazurkiewicz equivalence~\cite{mazurkiewicz1987,Abdulla14} 
and observation equivalence~\cite{Chalupa2017} are other models
that attempt to characterize causality.
Many of these models also incorporate the notion of last-write causality, similar to SHB.
However, these algorithms use expensive search algorithms like SAT solving
to explore the space of correct reorderings, unlike
a linear time vector clock algorithm like that for SHB.
Our experimental evaluation concurs with this observation.
Similar dependency relation called \emph{reads-from} is 
also used to characterize weak memory consistency 
semantics~\cite{Alglave:2014:HCM:2633904.2627752,HuangTSO16}. 

Race detection techniques can be broadly
classified as either being
static or dynamic.
Static race detection~\cite{Naik:2006:ESR:1133255.1134018,pratikakis11locksmith,Radoi:2013:PSR:2483760.2483765,racerx,voung2007relay,heisenbugs,Yahav:2001:VSP:373243.360206,echo}
is the problem of detecting if a program has an execution that exhibits
a data race, by analyzing its source code.
This problem, in its full generality, is undecidable
and practical tools employing static analysis techniques often face
a trade-off between scalability and precision.
Further, the use of such techniques often require the
programmer to add annotations to help guide static race detectors.

Dynamic race detection techniques, on the other hand,
examine a single execution of the
program to discover a data race in the program.
A large number of tools employing dynamic analysis are 
based on lockset-like analysis proposed by Eraser~\cite{savage1997eraser}.
Here, one tracks, for each memory location accessed, the set of
locks that protect the memory location on each access.
If this lockset becomes empty during the program execution,
a warning is issued.
Lockset-based analysis suffers from false positives.
Other dynamic race detectors  employ happens-before~\cite{lamport1978time} based
analysis. 
These include the use of 
vector clock~\cite{Mattern1988,Fidge:1991:LTD:112827.112860} 
algorithms such as \textsc{Djit}$^+$~\cite{Pozniansky:2003:EOD:966049.781529} and \fasttrack~\cite{fasttrack}
and the use of sets of threads and locks, as in, GoldiLocks~\cite{elmas2007goldilocks}.
As demonstrated in this paper, happens-before based analysis is 
sound only if limited to detecting the first race.
Other techniques can be categorized as \emph{predictive}
and can detect races missed by HB by exploring more correct reorderings of an observed trace.
These include use of SMT solvers~\cite{Said2011,rv2014,ipa2016,Huang2016}
and other techniques based on weakening the HB partial order
including CP~\cite{cp2012} and WCP~\cite{wcp2017}.
Amongst these, WCP is the only technique that has a linear running
time and is known to scale to large traces.
The soundness guarantee of partial order based techniques, like WCP and CP,
is again, limited to the first race.
Nevertheless, they do detect subtle races that HB can miss.
Our approach complements this line of research.
Other dynamic techniques such as random 
testing~\cite{Sen:2008:RDR:1375581.1375584},
sampling~\cite{marino2009literace,Erickson:2010:EDD:1924943.1924954},
and hybrid race detection~\cite{choi:2003:HDD:781498.781528}
are based on both locksets and happens-before relation.

\section{Conclusion}
\label{sec:conclusion}
%!TEX root = main.tex

Happens-before is a powerful technique that can be used to effectively
detect for races.  However, the detection power of HB is limited only
until the first race is identified.  We characterize when an HB-race,
beyond just the first race, can be scheduled in an alternate
reordering, by introducing a new partial order called SHB which
identifies all HB-schedulable races. SHB can be implemented in a
vector clock algorithm, which is only slightly different from HB
vector clock algorithm, and thus, existing race detection tools can
easily incorporate it to enhance their race detection capability.
Also, standard epoch like optimizations can be employed to improve the
performance of the basic algorithm.  We show, through extensive
experimental evaluation, the value our approach adds to sound race
detection tools.

In the future, we would like to extend the work for weaker partial
orders like CP~\cite{cp2012} and WCP~\cite{wcp2017}. Incorporating
control flow and data flow information in the traces is another
promising direction.

%Our characterization of read-write races is not exact. It would be
%interesting, from a mathematical standpoint, to obtain such a
%characterization which will need to overcome the challenges posed by
%the examples in Appendix~\ref{app:shb-examples}. Other directions for
%future work include a similar analysis for weaker partial orders like
%CP~\cite{cp2012} and WCP~\cite{wcp2017}. Incorporating control flow
%and data flow information in the traces can be a promising direction.

%% Acknowledgments
\begin{acks}                            %% acks environment is optional
                                        %% contents suppressed with 'anonymous'
  %% Commands \grantsponsor{<sponsorID>}{<name>}{<url>} and
  %% \grantnum[<url>]{<sponsorID>}{<number>} should be used to
  %% acknowledge financial support and will be used by metadata
  %% extraction tools.
  % This material is based upon work supported by the
  % \grantsponsor{GS100000001}{National Science
  %   Foundation}{http://dx.doi.org/10.13039/100000001} under Grant
  % No.~\grantnum{GS100000001}{nnnnnnn} and Grant
  % No.~\grantnum{GS100000001}{mmmmmmm}.  Any opinions, findings, and
  % conclusions or recommendations expressed in this material are those
  % of the author and do not necessarily reflect the views of the
  % National Science Foundation.
  We gratefully acknowledge National Science Foundation for supporting
 Umang Mathur (grant NSF CSR 1422798) and Mahesh Viswanathan (NSF CPS 1329991).
\end{acks}

\clearpage

\bibliography{references}

%% Bibliography

\clearpage

%% Appendix
\appendix
\section{Proof of Theorem~\ref{thm:SHBSoundness}}
\label{app:shb-proof}
%!TEX root = main.tex

In this section, we prove Theorem~\ref{thm:SHBSoundness}. We begin
with a couple of technical lemmas.

\begin{lemma}\label{lem:cr_hb}
% For any trace $\tr$, consider a correct reordering $\tr'$ of $\tr$
% that also respects $\hb{\tr}$. Then $\tr'$ respects $\shb{\tr}$, i.e.,
% for any $e,e'$ such that $e \shb{\tr} e'$ and $e' \in \events{\tr'}$,
% we have $e \in \events{\tr'}$ and $e \trord{\tr'} e'$.
Let $\tr$ be a trace and $\tr'$ be a correct reordering of $\tr'$ that
respects $\hb{\tr}$.  For any $e,e'$ such that $e \shb{\tr} e'$, if
$e' \in \events{\tr'}$ and $e'$ is not the last read event of its
thread in $\tr'$, then $e \in \events{\tr'}$ and $e \trord{\tr'} e'$.
\end{lemma}

\begin{proof}
Consider any $e,e'$ such that $e \shb{\tr} e'$, $e' \in \events{\tr'}$
and $e' = \ev{t, op}$ is not the last read event of the thread $t$ in
the trace $\tr'$.  Then it follows from Definition~\ref{def:shb} that
there is a sequence $e = e_0, e_1, \ldots e_n = e'$ such that for
every $i \leq n-1$, $e_i \trord{\tr} e_{i+1}$ and either (a)
$e_i \tho{\tr} e_{i+1}$ or (b) $e_i = \ev{t_i,\rel(\lk)}$, $e_{i+1}
= \ev{t_{i+1},\acq(\lk)}$, or (c) $e_{i+1} \in \reads{\tr}$ and $e_i
= \lw{\tr}(e_{i+1})$.

We will prove by induction on $i$, starting from $i = n$, that
$e_i \in \events{\tr'}$ and $e_i$ is not the last read event of its
thread in $\tr'$. Observe that these properties hold for $e' = e_n$
--- $e_n \in \events{\tr'}$ and $e_n$ is not the last read event of
its thread in $\tr'$. Assume we have established the claim for
$e_{i+1}$. Now there are three cases to consider for $e_i$. If
$e_i \tho{\tr'} e_{i+1}$ then clearly $e_i \in \events{\tr'}$ because
$e_{i+1} \in \events{\tr'}$. Further, if $e_i$ is a read event, then
it is not the last event of its thread because $e_{i+1}$ is after
it. If $e_i = \ev{t_i,\rel(\lk)}$ and $e_{i+1}
= \ev{t_{i+1},\acq(\lk)}$ then $e_i \in \events{\tr'}$ because $\tr'$
respects $\hb{\tr}$. Further $e_i$ is not the last read event because
it is not a read event! The last case to consider is where $e_i
= \lw{\tr}(e_{i+1})$. In this case, by induction hypothesis, we know
that $e_{i+1}$ is not the last read event of its thread, and therefore
by properties of a correct reordering, we have
$e_i \in \events{\tr'}$. Notice that in this case $e_i$ is not a read
event, and so the claim holds. Thus, we have established that $e =
e_0 \in \events{\tr'}$.

Next, we show that for every $i \leq n-1$, $e_i \trord{\tr'}
e_{i+1}$. If $e_i \tho{\tr} e_{i+1}$ or $e_i = \ev{t_i,\rel(\lk)}$ and
$e_{i+1} = \ev{t_{i+1},\acq(\lk)}$ with $e_i \trord{\tr} e_{i+1}$ then
$e_i \trord{\tr'} e_{i+1}$ because $\tr'$ respects $\hb{\tr}$. On the
other hand, if $e_i = \lw{\tr}(e_{i+1})$ then because $\tr'$ is a
correct reordering of $\tr$ and $e_{i+1}$ is not the last read event
of its thread (established in the previous paragraph), we have $e_i
= \lw{\tr}(e_{i+1}) = \lw{\tr'}(e_{i+1})$.  This establishes the fact
that $e = e_0 \trord{\tr'} e_n = e'$, which completes the proof of the
lemma.
\end{proof}

A slightly weaker form of the converse of Lemma~\ref{lem:cr_hb} also
holds.
\begin{lemma}\label{lem:shb_cr_hb}
For a trace $\tr$, let $\tr'$ be a trace with
$\events{\tr'} \subseteq \events{\tr}$ such that (a) $\tr'$ is
$\shb{\tr}$ downward closed, i.e., for any $e,e' \in \events{\tr}$ if
$e \shb{\tr} e'$ and $e' \in \events{\tr'}$ then
$e \in \events{\tr'}$, and (b) $\trord{\tr'} = \trord{\tr} \cap
(\events{\tr'} \times \events{\tr'})$. Then $\tr'$ is a correct
reordering of $\tr$ that respects $\hb{\tr}$.  Further,
for \textbf{every} read event $e \in \reads{\tr'}$, we have
$\lw{\tr'}(e) \simeq \lw{\tr}(e)$, i.e., either both $\lw{\tr'}(e)$
and $\lw{\tr}(e)$ are undefined, or they are both defined and equal.
\end{lemma}

\begin{proof}
The trace $\tr'$ in the lemma is such that the events in $\tr'$ are
downward closed with respect to $\shb{\tr}$ and in $\tr'$ they are
ordered in exactly the same way as in $\tr$. The fact that $\tr'$
respects $\hb{\tr}$ simply follows from the fact that
$\hb{\tr} \subseteq \shb{\tr}$ and
$\hb{\tr} \subseteq \trord{\tr}$. So the main goal is to establish
that $\tr'$ is a correct reordering of $\tr$ that preserves the last
writes of \emph{all} read events.

First we show that $\tr'$ respects lock semantics. Suppose $e_1
= \ev{t_1, \acq(\lk)}$ and $e_2 = \ev{t_2, \acq(\lk)}$ are two lock
acquire events for some lock $\lk$ such that $e_1 \trord{\tr} e_2$ and
$\{e_1,e_2\} \subseteq \events{\tr'}$. Let $e_1'$ be the matching
release event for $e_1$ in $\tr$; such an $e_1'$ exists because $\tr$
is a valid trace. Then we have $e_1 \hb{\tr} e_1' \hb{\tr} e_2$, and
so $e_1' \in \events{\tr'}$ and $e_1' \trord{\tr'} e_2$ because $\tr'$
respects $\hb{\tr}$.

Next observe that since $\tho{\tr} \subseteq \hb{\tr}$ and $\tr'$
respects $\hb{\tr}$, we can conclude that $\proj{\tr'}{t}$ is a prefix
of $\proj{\tr}{t}$ for any thread $t$.

Finally, consider any $e' \in \reads{\tr'}$. Suppose $\lw{\tr}(e')$ is
defined. Let $e = \lw{\tr}(e')$. Since $e \shb{\tr} e'$ and $\tr'$ is
downward closed with respect to $\shb{\tr}$, we have
$e \in \events{\tr'}$. Let $e_1 = \lw{\tr'}(e')$. We need to argue
that $e_1 = e$. Suppose (for contradiction) it is not, i.e., $e \neq
e_1$. Then either $e_1 \trord{\tr} e$ or $e' \trord{\tr} e_1$, because
$e = \lw{\tr}(e')$. However, the fact that $e_1 = \lw{\tr'}(e')$
contradicts the fact that $\trord{\tr'} = \trord{\tr} \cap
(\events{\tr'} \times \events{\tr'})$. Conversely, if $\lw{\tr'}(e')$
is defined then let $e = \lw{\tr'}(e')$. Since $\trord{\tr'}
= \trord{\tr} \cap (\events{\tr'} \times \events{\tr'})$, we have
$e \trord{\tr} e'$. Thus, $\lw{\tr}(e')$ is defined. Let $e_1
= \lw{\tr}(e')$. Once again, since $e_1 \shb{\tr} e'$, and $\tr'$ is
downward closed with respect to $\shb{\tr}$, we have
$e_1 \in \tr'$. Just like in the previous direction, we can conclude
that $e = e_1$ because otherwise we violate the fact that
$\trord{\tr'}$ is identical to $\trord{\tr}$ over $\events{\tr'}$.
\end{proof}

We now prove Theorem~\ref{thm:SHBSoundness} below
\begin{reptheorem}{thm:SHBSoundness}
Let $\tr$ be a trace and $e_1 \trord{\tr} e_2$ be conflicting events in $\tr$.
$(e_1, e_2)$ is an $\hb{\tr}$-schedulable race iff 
either $\ltho{\tr}(e_2)$ is undefined, or $e_1 \not\leq^\tr_{\mathsf{SHB}} \ltho{\tr}(e_2)$.
\end{reptheorem}

\begin{proof}
Let us first prove the forward direction.
That is, let $(e_1,e_2)$ be an HB-race such that the
event $e = \ltho{\tr}(e_2)$ is defined and $e_1 \shb{\tr} e$. 
Consider any correct reordering $\tr'$ that contains both
$e_1$ and $e_2$ and respects $\hb{\tr}$. 
First, since $\tr'$ is a correct reordering of $\tr$, we must have
$e \in \events{\tr'}$ and $e \trord{\tr'} e_2$.
Further, since $e_1 \shb{\tr} e$, from Lemma~\ref{lem:cr_hb}, 
$e_1 \trord{\tr'} e$.
Thus, we have that $e_1 \trord{\tr'} e \trord{\tr'} e_1$
for any correct reordering $\tr'$ of $\tr$ that respects $\hb{\tr}$.
This means, $(e_1, e_2)$ cannot be a $\hb{\tr}$-schedulable race.
% Also, since $\tr'$ is a correct reordering
% also respects $\shb{\tr}$. But that would mean that
% $e \in \events{\tr'}$ and $e_1 \trord{\tr'} e \trord{\tr'} e_2$, and
% so $e_1,e_2$ are not consecutive in $\tr'$. Since this holds for any
% correct reordering $\tr'$ of $\tr$ that also respects $\hb{\tr}$, we
% can conclude that $(e_1,e_2)$ is not an $\hb{\tr}$-schedulable race.

We now prove the backward direction.
Consider an HB-race $(e_1,e_2)$ such that 
either $\ltho{\tr}(e_2)$ is undefined, or if it exists, then it satisfies 
$e_1 \not\leq^\tr_{\mathsf{SHB}} \ltho{\tr}(e_2)$.
% satisfying either of the two
% conditions. 
% Notice that in either case this means that there is no
% event $e \in \events{\tr} \setminus \{e_1,e_2\}$ such that
% $e_1 \shb{\tr} e \shb{\tr} e_2$. 
Consider the set $\setreq$ defined as
\[
\setreq = \{e \in \events{\tr} \setminus\{e_1,e_2\}\: |\: 
        e \shb{\tr} e_1 \mbox{ or } 
        e \shb{\tr} \ltho{\tr}(e_2)\}
\]
where we assume that if $\ltho{\tr}(e_2)$ is undefined then no event
$e$ satisfies the condition $e \shb{\tr} \ltho{\tr}(e_2)$.

First we will show that $\setreq$ is downward closed with respect to
$\shb{\tr}$. Consider $e,e'$ such that $e \shb{\tr} e'$ and
$e' \in \setreq$. By definition of $\setreq$, we have
$e' \not\in \{e_1,e_2\}$ and either $e' \shb{\tr} e_1$ or
$e' \shb{\tr} \ltho{\tr}(e_2)$. Observe that if
$e \not\in \{e_1,e_2\}$, then it is clear that $e \in \setreq$ by
definition since $\shb{\tr}$ is transitive. It is easy to see that
$e \neq e_2$ --- this is because since $e' \neq e_2$, and
$\shb{\tr} \subseteq \trord{\tr}$, $e' \stricttrord{\tr} e_2$ and so
$e \stricttrord{\tr} e_2$. So, all we have left to establish is that
$e \neq e_1$. Suppose for contradiction $e = e_1$. Then it must be the
case that $e' \shb{\tr} \ltho{\tr}(e_2)$. Since $e_1 = e \shb{\tr}
e' \shb{\tr} \ltho{\tr}(e_2)$, we have
$e_1 \shb{\tr} \ltho{\tr}(e_2)$, which contradicts our assumption
about $(e_1,e_2)$.

Let us now consider a trace $\tr''$ which consists of the events in
$\setreq$ ordered according to $\trord{\tr}$. That is, $\trord{\tr''}
= \trord{\tr} \cap (\setreq \times \setreq)$. Since $\tr''$ satisfies the
conditions of Lemma~\ref{lem:shb_cr_hb}, we can conclude that $\tr''$
is a correct reordering of $\tr$ that respects $\hb{\tr}$ and preserves the
last-writes of \emph{every} read event present.

Consider the trace $\tr' = \tr''e_1e_2$. First we prove that $\tr'$
respects $\hb{\tr}$. To do that, we first show that for any event
$e\in \events{\tr}$ such that $e \hb{\tr} e_1$ and $e \neq e_1$, or
$e \hb{\tr} e_2$ and $e \neq e_2$, then $e \in \setreq$.  If
$e \hb{\tr} e_1$ then $e \shb{\tr} e_1$ and so $e \in \setreq$. On the
other hand, if $e \hb{\tr} e_2$ (and $e \neq e_2$), since $(e_1,e_2)$
is an HB-race, we must have $e \neq e_1$ and
$e \hb{\tr} \ltho{\tr}(e_2)$. So $e \in \setreq$. Now the fact $\tr'$
respects $\hb{\tr}$ follows from the fact that $\tr''$ respects
$\hb{\tr}$ and the claim just proved.

We now prove that $\tr'$ is a correct reordering. Observe that since
$\tr'$ respects $\hb{\tr}$, $\tr'$ is well formed (lock semantics is
not violated) and preserves thread-wise prefixes ($\forall
t, \proj{\tr'}{t}$ is a prefix of $\proj{\tr}{t}$).  Further, $\tr''$
is such that every read event in $\tr''$ reads the same last write as
in $\tr$.  Also, since $e_1$ and $e_2$ are the last events in their
threads in $\tr'$, we conclude that $\tr'$ is a correct reordering of
$\tr$ that respects $\hb{\tr}$.
\end{proof}

\section{Proofs for Algorithm~\ref{algo:vc}}
\label{app:algo}
%!TEX root = main.tex

% \subsection{Proofs for Algorithm~\ref{algo:vc}}
% \label{app:algo_update}

We now prove Theorem~\ref{thm:isomorphicVC}, which states the
correctness of Algorithm~\ref{algo:vc}. Before establishing this claim
we would like to introduce some notation and prove some auxiliary
claims. 

Let us fix a trace $\tr$. Recall that for any event $e$, $C_e$ is the
(vector) timestamp assigned by Algorithm~\ref{algo:vc}. Let us denote
by $L_\lk^e$ the value of clock $\Ll_\lk$ just before the event $e$ is
processed. Similarly, let $LW_x^e$ denote the value of clock $\LW_x$
just before $e$ is processed. It is easy to see that the following
invariant is maintained by Algorithm~\ref{algo:vc}.
\begin{proposition}
\label{prop:vc-invariant1}
Let $e$ be an arbitrary event of trace $\tr$. Let $e_\lk$ be the last
$\rel(\lk)$-event in $\tr$ before $e$, and let $e_x$ be the last
$\wt(x)$-event in $\tr$ before $e$ (with respect to
$\trord{\tr}$). Note that $e_\lk$ and $e_x$ maybe undefined. Then,
$L_\lk^e = C_{e_\lk}$ and $LW_x^e = C_{e_x}$, where if an event $f$ is
undefined, we take $C_f = \bot$.
\end{proposition}
\begin{proof}
The observation follows from the way vector clocks $\Ll_\lk$ and
$\LW_x$ are updated.
\end{proof}

Another invariant that follows from the update rules of
Algorithm~\ref{algo:vc} is the following.
\begin{proposition}
\label{prop:vc-invariant2}
Let $e_1$ and $e_2$ be events of thread $t$ such that $e_1 \trord{\tr}
e_2$, i.e., $e_1 \tho{\tr} e_2$. Let $t'$ be any thread such that
$t \neq t'$. Then the following observations hold.
\begin{enumerate}
\item\label{lbl:monotonic} $C_{e_1} \cle C_{e_2}$.
\item\label{lbl:remote-clk-upd} $C_{e_1}(t') = C_{e_2}(t')$ unless
  there is an event $e$ of thread $t$ that is either an $\acq$-event, or
  a $\rd$-event, or a $\join$-event such that $e \neq e_1$ and
  $e_1 \trord{\tr} e \trord{\tr} e_2$.
\item\label{lbl:local-clk-upd} $C_{e_1}(t) = C_{e_2}(t)$ unless there
  is an event $e$ of thread $t$ that is either a $\rel$-event, or a
  $\wt$-event, or a $\fork$-event such that $e \neq e_2$ and
  $e_1 \trord{\tr} e \trord{\tr} e_2$; in this case $C_{e_1}(t) <
  C_{e_2}(t)$.
\end{enumerate}
\end{proposition}
\begin{proof}
Follows from the way $\Cc_t$ is updated by Algorithm~\ref{algo:vc}.
\end{proof}

We now prove the main lemma crucial to the correctness of
Algorithm~\ref{algo:vc}, that relates $\shb{}$ to the ordering on
vector clocks.
\begin{lemma}
\label{lem:vc-shb}
Let $e = \ev{t,op}$ be an event such that $C_e(t') = k$ for some $t'
\neq t$. Let $e' = \ev{t',op'}$ be the last event such that
$C_{e'}(t') = k$. Then $e' \shb{\tr} e$. 
\end{lemma}
\begin{proof}
The result will be proved by induction on the position of $e$ in the
trace $\tr$. Observe that if $e$ is the first event of $\tr$, then
$C_e(t') = 0$ for all $t' \neq t$, no matter what event $e$ is. And
there is no event $e' = \ev{t',op'}$ such that $C_{e'}(t') = 0$. Thus,
the lemma holds vaccuously in the base case.

Let us now consider the inductive step. Define $e_1 = \ev{t,op_1}$ be
the last event in $\tr$ before $e$ (possibly same as $e$) such that
$op_1$ is either $\acq$, $\rd$, or $\join$; if no such $e_1$ exists
then take $e_1$ to be the first event performed by $t$.  Notice, by
our choice of $e_1$ and
Proposition~\ref{prop:vc-invariant2}(\ref{lbl:remote-clk-upd}), for
every $t'' \neq t$, $C_{e_1}(t'') = C_e(t'')$. If $e_1 \neq e$, the
result follows by induction hypothesis on $e_1$. 

Let us assume $e_1 = e$. We need to consider different cases based on
what $e_1$ is.
\begin{itemize}
\item {\bf Case $e = e_1 = \ev{t,\acq(\lk)}$:} Let $f_1$ be the event 
  immediately before $e$ in $\proj{\tr}{t}$ and $f_2$ be the event
  such that $C_{f_2} = L_\lk^e$ (given by
  Proposition~\ref{prop:vc-invariant1}). Note that both $f_1$ and
  $f_2$ may be undefined. Also notice that, for any $t'$, either
  $C_e(t') = 0$, or $C_e(t') = C_{f_1}(t') \neq 0$ (and $f_1$ is
  defined), or $C_e(t') = C_{f_2}(t') \neq 0$ (and $f_2$ is
  defined). If $C_e(t') = 0$ then the lemma follows vaccuously as in
  the base case because there is no event $e' = \ev{t',op'}$ with
  $C_{e'}(t') = 0$. Let us now consider the remaining cases. Let $t_2$
  denote the thread performing $f_2$, if $f_2$ is defined. Consider
  the case when either $C_e(t') = C_{f_1}(t') \neq 0$ or $C_e(t') =
  C_{f_2}(t')$ with $t' \neq t_2$. In this situation, the lemma
  follows using the induction hypothesis on either $f_1$ or $f_2$
  since both $f_1$ and $f_2$ (when defined) are $\shb{\tr} e$. The
  last case to consider is when $t' = t_2$ and $C_e(t') =
  C_{f_2}(t')$. By
  Proposition~\ref{prop:vc-invariant2}(\ref{lbl:local-clk-upd}), $f_2$
  is the last event of $t' = t_2$ whose $t'$th component is
  $k$. Further, by definition $f_2 \shb{\tr} e$, and so the lemma
  holds.
\item {\bf Case $e = e_1 = \ev{t,\rd(x)}$:} Let $f_1$ be the event
  immediately before $e$ in $\proj{\tr}{t}$ and $f_2$ be the event
  such that $C_{f_2} = LW_x^e$ (given by
  Proposition~\ref{prop:vc-invariant1}). Again, both $f_1$ and $f_2$
  may be undefined. Also notice that, for any $t'$, either $C_e(t') =
  0$, or $C_e(t') = C_{f_1}(t') \neq 0$ (and $f_1$ is defined), or
  $C_e(t') = C_{f_2}(t') \neq 0$ (and $f_2$ is defined). If $C_e(t') =
  0$ then the lemma follows vaccuously as in the base case because
  there is no event $e' = \ev{t',op'}$ with $C_{e'}(t') = 0$. Let us
  now consider the remaining cases. Let $t_2$ denote the thread
  performing $f_2$, if $f_2$ is defined. Consider the case when either
  $C_e(t') = C_{f_1}(t') \neq 0$ or $C_e(t') = C_{f_2}(t')$ with
  $t' \neq t_2$. In this situation, the lemma follows using the
  induction hypothesis on either $f_1$ or $f_2$ since both $f_1$ and
  $f_2$ (when defined) are $\shb{\tr} e$. The last case to consider is
  when $t' = t_2$ and $C_e(t') = C_{f_2}(t')$. By
  Proposition~\ref{prop:vc-invariant2}(\ref{lbl:local-clk-upd}), $f_2$
  is the last event of $t' = t_2$ whose $t'$th component is
  $k$. Further, by definition $f_2 \shb{\tr} e$, and so the lemma
  holds.
\item {\bf Case $e = e_1 = \ev{t,\join(t_1)}$:} Let $f_1$ be the event
  immediately before $e$ in $\proj{\tr}{t}$ and $f_2$ be the last
  event of the form $\ev{t_1,op}$. Again, both $f_1$ and $f_2$ may be
  undefined. Also notice that, for any $t'$, either (a) $C_e(t') = 0$,
  or (b) $t_1 = t'$, $C_e(t') = 1$, and $f_2$ is undefined, or (c)
  $C_e(t') = C_{f_1}(t') \neq 0$ and $f_1$ is defined, or (d) $C_e(t')
  = C_{f_2}(t') \neq 0$ and $f_2$ is defined. In cases (a) or (b)
  above, the lemma follows vaccuously as in the base case because
  there is no event $e' = \ev{t',op'}$ with $C_{e'}(t') = k$ (where
  $k$ is either $0$ or $1$ depending on which we case we
  consider). Let us now consider the remaining cases. Let $t_2$ denote
  the thread performing $f_2$, if $f_2$ is defined. Consider the case
  when either $C_e(t') = C_{f_1}(t') \neq 0$ or $C_e(t') =
  C_{f_2}(t')$ with $t' \neq t_2$ (and $f_2$ defined). In this
  situation, the lemma follows using the induction hypothesis on
  either $f_1$ or $f_2$ since both $f_1$ and $f_2$ (when defined) are
  $\shb{\tr} e$. The last case to consider is when $t' = t_2$ and
  $C_e(t') = C_{f_2}(t')$. By definition, $f_2$ is the last event of
  $t' = t_2$ whose $t'$th component is $k$. Further, by definition
  $f_2 \shb{\tr} e$, and so the lemma holds.
\item {\bf Case $e = e_1$ is the first event:} This is the case when 
  the above 3 cases don't hold. So $e = e_1$ is not an $\acq$-event,
  nor a $\rd$-event, nor a $\join$-event. Moreover, since $e$ is the
  first event of thread $t$ and is of the form $\ev{t,op}$, it must be
  the the thread $t$ has not been forked by any thread in $\tr$. Thus,
  for any $t' \neq t$, $C_e(t') = 0$. The lemma, therefore, follows
  vaccuously as in the base case. \qedhere
\end{itemize}
\end{proof}

We are ready to present the proof of Theorem~\ref{thm:isomorphicVC}.

\begin{reptheorem}{thm:isomorphicVC}
For events $e, e' \in \events{\tr}$ such that $e \trord{\tr} e'$,
$C_e \cle C_{e'}$ iff $e \shb{\tr} e'$
\end{reptheorem}

\begin{proof}
Let us first prove the implication from left to right. Consider $e,e'$
such that $e \trord{\tr} e'$. If $e \tho{\tr} e'$ then $e \shb{\tr}
e'$ since $\tho{\tr} \subseteq \shb{\tr}$. On the other hand, if $e$
and $e'$ are not events of the same thread, then this direction of the
theorem follows from Lemma~\ref{lem:vc-shb}.

Let us now prove the implication from right to left. Consider events
such that $e \shb{\tr} e'$. Then, by definition, we have a sequence of
events $e = f_1,f_2,\ldots f_k = e'$ such that for every $i$,
$f_i \trord{\tr} f_{i+1}$ and either (i) $f_i$ and $f_{i+1}$ are both
events of the form $\ev{t,op}$, or (ii) $f_i$ is a $\rel(\lk)$-event
and $f_{i+1}$ is a $\acq(\lk)$-event, or (iii) $f_i$ is a
$\fork(t)$-event and $f_{i+1}$ is an event of the form $\ev{t,op}$, or
(iv) $f_i$ is an event of the form $\ev{t,op}$ and $f_{i+1}$ is a
$\join(t)$-event, or (v) $f_i = \lw{\tr}(f_{i+1})$. In each of these
cases, Algorithm~\ref{algo:vc} ensures that $C_{f_i} \cle
C_{f_{i+1}}$. Thus, we have $C_e \cle C_{e'}$.
\end{proof}

We now prove Theorem~\ref{thm:correct-races}.  We first note some
auxiliary propositions.  Let us denote by $R_x^e$ the value of clock
$\Rr_x$ just before the event $e$ is processed.  Similarly, let
$W_x^e$ denote the value of clock $\Ww_x$ just before $e$ is
processed. It is easy to see that the following invariant is
maintained by Algorithm~\ref{algo:vc}.

\begin{proposition}
\label{prop:vc-invariantrw}
Let $e$ be an arbitrary event of trace $\tr$. Let $e_t^{\rd(x)}$ 
be the last $\ev{t, \rd(x)}$-event in $\tr$ before $e$, 
and let $e_t^{\wt(x)}$ be the last
$\ev{t, \wt(x)}$-event in $\tr$ before $e$ (with respect to
$\trord{\tr}$). 
Note that $e_t^{\rd(x)}$  and $e_t^{\wt(x)}$  maybe undefined. 
Then, $\forall t, R_x(t) = C_{e_t^{\rd(x)}}(t)$ and 
$\forall t, W_x(t) = C_{e_t^{\wt(x)}}(t)$ 
where if an event $f$ is undefined, we take $C_f = \bot$.
\end{proposition}

\begin{proof}
The observation follows from the way vector clocks $\Rr_x$ and
$\Ww_x$ are updated.
\end{proof}

\begin{lemma}
\label{lem:prev_clock}
Let $e_1, e_2 \in \events{\tr}$ performed by threads $t_1$ and $t_2$,
respectively, such that $t_1 \neq t_2$.  Then, $e_1 \shb{\tr} e_2$ iff
$C_{e_1} \cle C_{e_2}[( C_{e_2}(t_2) + 1)/t_2]$.
\end{lemma}

\begin{proof}
Let $c_2 = C_{e_2}(t_2)$.
First suppose that $e_1 \shb{\tr} e_2$.
Then, from Theorem~\ref{thm:isomorphicVC}, we have $C_{e_1} \cle C_{e_2}$
and thus $C_{e_1} \cle C_{e_2}[(c_2+1)/t_2]$.
Next, assume that $C_{e_1} \cle C_{e_2}[(c_2 + 1)/t_2]$.
In particular, $C_{e_1}(t_1) \leq C_{e_2}(t_1)$.
Then by Lemma~\ref{lem:vc-shb}, we have $e_1 \shb{\tr} e_2$
\end{proof}

\begin{reptheorem}{thm:correct-races}
Let $e$ be a read/write event $e \in \events{\tr}$.
Algorithm~\ref{algo:vc} reports a race at $e$
iff there is an event $e' \in \events{\tr}$ such that $(e', e)$
is an $\hb{\tr}$-schedulable race.
\end{reptheorem}

\begin{proof}
% \ucomment{Prove this. You will have to prove that $\Rr_x$ and $\Ww_x$ are maintained correctly.}
% Let $e'$ and $e$ be conflicting events in $\tr$.
% Then by Theorem~\ref{thm:SHBSoundness}, $(e, e')$ is an $\hb{\tr}$-schedulable race
% iff either $\ltho{\tr}(e)$ does not exist or if it exists, then 
% $e' \not\leq_\textsf{SHB}^\tr \ltho{\tr}(e)$.
% 
Let us first consider the case when $\ltho{\tr}(e)$ is not defined.
Then, the value of the clock $\Cc_t = \bot[1/t]$ at line 19, 24 or 26 (depending
upon whether $e$ is a read or a write event).
If the check $\neg(\Ww_x \cle \Cc_t)$ passes, then
there is a $t'$ such that $\Ww_x(t') > \Cc_t$ and thus
the there is an event $e'$ (namely the last write event of $x$ in thread $t'$)
that conflicts with $e$. Thus, $(e', e)$ is a $\hb{\tr}$-schedulable race
by Theorem~\ref{thm:SHBSoundness}.
On the other hand, if the check fails, then
$\Ww_x = \bot[1/t]$ or $\Ww_x = \bot$ and in either case
there is no event that conflicts with $e$.
One can similarly argue that the checks on Lines 24 and 26
are both necessary and sufficient for the case when $e$ is a write event
and $\ltho{\tr}(e)$ is undefined.
% Similarly, for a write event $e$, if one of the checks
% $\neg(\Rr_x \cle \Cc_t)$ or $\neg(\Ww_x \cle \Cc_t)$ passes,
% it means that there is an event $e'$ that conflicts with $e$
% and thus $(e', e)$ is an $\hb{\tr}$-schedulable race.

Next we consider the case when $f = \ltho{\tr}(e)$ is defined.  Now
let $e$ be a read event.  If the check $\neg(\Ww_x \cle \Cc_t)$
passes, then there is a $t'$ such that $\Ww_x(t') > \Cc_t(t')$ and
thus there is an event $e' = \ev{t', \wt(x)}$ such that $C_{e'}(t')
> \Cc_t(t')$ and thus $C_{e'} \not\cle \Cc_t$; note that it must be
the case that $t' \neq t$.  Depending upon whether $f$ is a
read/join/acquire event or a write/fork/release event, the value of
the clock $\Cc_t$ at Line 19 is $\Cc_t = C_f$ or $\Cc_t =
C_f[(C_f(t)+1)/t]$.  In either case, by Lemma~\ref{lem:prev_clock}, we
have that $e' \not\leq_\textsf{SHB}^\tr f$.  On the other hand if,
$\Ww_x \cle \Cc_t$, then $\forall t' \neq t,
C_{e_{t'}^{\wt(x)}}(t') \leq \Cc_t(t')$, where $e_{u}^{\wt(x)}$ is the
last write event of $x$ performed by thread $u$.  This means that for
every event $e'$ such that $e'$ conflicts with, by
Lemma~\ref{lem:prev_clock}, we have $e' \shb{\tr} f$ and thus $(e',
e)$ is not an $\hb{\tr}$-schedulable race.  The argument for the case
when $e$ is a write event is similar.
\end{proof}

We now establish the asymptotic space and time bounds for
Algorithm~\ref{algo:vc}.

\begin{reptheorem}{thm:complexityVC}
For a trace $\tr$ with $n$ events, $T$ threads, $V$ variables, and $L$
locks, Algorithm~\ref{algo:vc} runs in time $O(nT\log n)$ and uses
$O((V+L+T)T\log n)$ space.
\end{reptheorem}

\begin{proof}
Observe that for a trace of length $n$, every component of each of the
vector clocks is bounded by $n$. Thus, each vector clock takes space
$O(T\log n)$, where $T$ is the number of threads. We have a vector
clock for each thread, lock, and variable, which gives us a space
bound of $O((V+T+L)T\log n)$. Notice that to process any event we need
to update constantly many vector clocks. The time to update any single
vector clock can be bounded by its size $O(T\log n)$. Thus, the total
running time is $O(nT\log n)$.
\end{proof}

\section{False Positives Reported by Existing Practical Dynamic Race Detection Tools}
\label{app:false_races}
%!TEX root = main.tex

We evaluate existing tools that
use happens-before based race detection to
check if they report false (unschedulable) races.

We use the following program in Figure~\ref{fig:example1} to assess if these
tools guarantee soundness.
The Java source code for this program is in Figure~\ref{simple_java}
and the C source code is in Figure~\ref{simple_c}

%!TEX root = main.tex

\begin{figure}[t]
\begin{subfigure}{.4\textwidth}
\begin{Verbatim}[commandchars=\\\{\}, fontsize=\footnotesize, baselinestretch=0.5]
\PYG{k+kd}{public} \PYG{k+kd}{class} \PYG{n+nc}{Test} \PYG{k+kd}{extends} \PYG{n}{Thread}\PYG{o}{\PYGZob{}}

    \PYG{k+kd}{static} \PYG{k+kt}{int} \PYG{n}{x}\PYG{o}{,} \PYG{n}{y}\PYG{o}{;}
    \PYG{k+kd}{public} \PYG{k+kt}{int} \PYG{n}{id}\PYG{o}{;}

    \PYG{n+nd}{@Override}
    \PYG{k+kd}{public} \PYG{k+kt}{void} \PYG{n+nf}{run}\PYG{o}{()} \PYG{o}{\PYGZob{}}
        \PYG{k}{if}\PYG{o}{(}\PYG{n}{id} \PYG{o}{==} \PYG{l+m+mi}{1}\PYG{o}{)\PYGZob{}}
            \PYG{n}{y} \PYG{o}{=} \PYG{n}{x} \PYG{o}{+} \PYG{l+m+mi}{5}\PYG{o}{;}
        \PYG{o}{\PYGZcb{}}
        \PYG{k}{if}\PYG{o}{(}\PYG{n}{id} \PYG{o}{==} \PYG{l+m+mi}{2}\PYG{o}{)\PYGZob{}}
            \PYG{k}{if} \PYG{o}{(} \PYG{n}{y}\PYG{o}{==}\PYG{l+m+mi}{5} \PYG{o}{)\PYGZob{}}
                \PYG{n}{x} \PYG{o}{=} \PYG{l+m+mi}{10}\PYG{o}{;}
            \PYG{o}{\PYGZcb{}}
            \PYG{k}{else}\PYG{o}{\PYGZob{}}
                \PYG{k}{while}\PYG{o}{(}\PYG{k+kc}{true}\PYG{o}{);}
            \PYG{o}{\PYGZcb{}}
        \PYG{o}{\PYGZcb{}}
    \PYG{o}{\PYGZcb{}}

    \PYG{k+kd}{public} \PYG{k+kd}{static} \PYG{k+kt}{void} \PYG{n+nf}{main}\PYG{o}{(}\PYG{n}{String} \PYG{n}{args}\PYG{o}{[])}
            \PYG{k+kd}{throws} \PYG{n}{Exception} \PYG{o}{\PYGZob{}}
        \PYG{k+kd}{final} \PYG{n}{Test} \PYG{n}{t1} \PYG{o}{=} \PYG{k}{new} \PYG{n}{Test}\PYG{o}{();}
        \PYG{k+kd}{final} \PYG{n}{Test} \PYG{n}{t2} \PYG{o}{=} \PYG{k}{new} \PYG{n}{Test}\PYG{o}{();}
        \PYG{n}{t1}\PYG{o}{.}\PYG{n+na}{id} \PYG{o}{=} \PYG{l+m+mi}{1}\PYG{o}{;}
        \PYG{n}{t2}\PYG{o}{.}\PYG{n+na}{id} \PYG{o}{=} \PYG{l+m+mi}{2}\PYG{o}{;}
        \PYG{n}{t1}\PYG{o}{.}\PYG{n+na}{start}\PYG{o}{();}
        \PYG{n}{t2}\PYG{o}{.}\PYG{n+na}{start}\PYG{o}{();}
        \PYG{n}{t1}\PYG{o}{.}\PYG{n+na}{join}\PYG{o}{();}
        \PYG{n}{t2}\PYG{o}{.}\PYG{n+na}{join}\PYG{o}{();}
    \PYG{o}{\PYGZcb{}}
\PYG{o}{\PYGZcb{}}
\end{Verbatim}
\vspace{0.5cm}
\caption{Multi-threaded Java program}
\label{simple_java}
\end{subfigure}
\hfill
\begin{subfigure}{.4\textwidth}
\begin{Verbatim}[commandchars=\\\{\}, fontsize=\footnotesize, baselinestretch=0.5]
\PYG{c+cp}{\PYGZsh{}include} \PYG{c+cpf}{\PYGZlt{}pthread.h\PYGZgt{}}
\PYG{c+cp}{\PYGZsh{}include} \PYG{c+cpf}{\PYGZlt{}stdio.h\PYGZgt{}}

\PYG{k+kt}{int} \PYG{n}{x}\PYG{p}{;}
\PYG{k+kt}{int} \PYG{n}{y}\PYG{p}{;}

\PYG{k+kt}{void} \PYG{o}{*}\PYG{n+nf}{Thread1}\PYG{p}{(}\PYG{k+kt}{void} \PYG{o}{*}\PYG{n}{a}\PYG{p}{)} \PYG{p}{\PYGZob{}}
  \PYG{n}{y} \PYG{o}{=} \PYG{n}{x} \PYG{o}{+} \PYG{l+m+mi}{5}\PYG{p}{;}
  \PYG{k}{return} \PYG{n+nb}{NULL}\PYG{p}{;}
\PYG{p}{\PYGZcb{}}

\PYG{k+kt}{void} \PYG{o}{*}\PYG{n+nf}{Thread2}\PYG{p}{(}\PYG{k+kt}{void} \PYG{o}{*}\PYG{n}{a}\PYG{p}{)} \PYG{p}{\PYGZob{}}
  \PYG{k}{if} \PYG{p}{(}\PYG{n}{y} \PYG{o}{==} \PYG{l+m+mi}{5}\PYG{p}{)\PYGZob{}}
    \PYG{n}{x} \PYG{o}{=} \PYG{l+m+mi}{10}\PYG{p}{;}
  \PYG{p}{\PYGZcb{}}
  \PYG{k}{else}\PYG{p}{\PYGZob{}}
    \PYG{k}{while}\PYG{p}{(}\PYG{n+nb}{true}\PYG{p}{)\PYGZob{}\PYGZcb{}}
  \PYG{p}{\PYGZcb{}}
  \PYG{k}{return} \PYG{n+nb}{NULL}\PYG{p}{;}
\PYG{p}{\PYGZcb{}}

\PYG{k+kt}{int} \PYG{n+nf}{main}\PYG{p}{()} \PYG{p}{\PYGZob{}}
  \PYG{n}{x} \PYG{o}{=} \PYG{l+m+mi}{0}\PYG{p}{;}
  \PYG{n}{y} \PYG{o}{=} \PYG{l+m+mi}{0}\PYG{p}{;}
  \PYG{n}{pthread\PYGZus{}t} \PYG{n}{t}\PYG{p}{[}\PYG{l+m+mi}{2}\PYG{p}{];}
  \PYG{n}{pthread\PYGZus{}create}\PYG{p}{(}\PYG{o}{\PYGZam{}}\PYG{n}{t}\PYG{p}{[}\PYG{l+m+mi}{0}\PYG{p}{],} \PYG{n+nb}{NULL}\PYG{p}{,} \PYG{n}{Thread1}\PYG{p}{,} \PYG{n+nb}{NULL}\PYG{p}{);}
  \PYG{n}{pthread\PYGZus{}create}\PYG{p}{(}\PYG{o}{\PYGZam{}}\PYG{n}{t}\PYG{p}{[}\PYG{l+m+mi}{1}\PYG{p}{],} \PYG{n+nb}{NULL}\PYG{p}{,} \PYG{n}{Thread2}\PYG{p}{,} \PYG{n+nb}{NULL}\PYG{p}{);}
  \PYG{n}{pthread\PYGZus{}join}\PYG{p}{(}\PYG{n}{t}\PYG{p}{[}\PYG{l+m+mi}{0}\PYG{p}{],} \PYG{n+nb}{NULL}\PYG{p}{);}
  \PYG{n}{pthread\PYGZus{}join}\PYG{p}{(}\PYG{n}{t}\PYG{p}{[}\PYG{l+m+mi}{1}\PYG{p}{],} \PYG{n+nb}{NULL}\PYG{p}{);}
\PYG{p}{\PYGZcb{}}
\end{Verbatim}
\vspace{0.5cm}
\caption{Multi-threaded C program}
\label{simple_c}
\end{subfigure}
\caption{Concurrent program from Figure~\ref{fig:program1}}
\label{fig:programs}
\end{figure}

\subsection{\fasttrack}

When run on \fasttrack, the Java program in Fig~\ref{simple_java}, \fasttrack~produces
the following output:

\vspace{0.5cm}
\hrule
{\small
\begin{verbatim}
## 
## =====================================================================
## HappensBefore Error
## 
##          Thread: 2    
##           Blame: Test.y_I
##           Count: 1    (max: 100)
##     Guard State: [(0:0) (1:1)]
##           Class: Test
##           Field: null.Test.y_I
##           Locks: []
##         Prev Op: write-by-thread-1
##      Prev Op CV: [(0:0) (1:1)]
##          Cur Op: read
##       Cur Op CV: [(0:1) (1:0) (2:1)]
##           Stack: Use -stacks to show stacks...
## =====================================================================
## 
## 
## =====================================================================
## HappensBefore Error
## 
##          Thread: 2    
##           Blame: Test.x_I
##           Count: 1    (max: 100)
##     Guard State: [(0:0) (1:1)]
##           Class: Test
##           Field: null.Test.x_I
##           Locks: []
##         Prev Op: read-by-thread-1
##      Prev Op CV: [(0:0) (1:1)]
##          Cur Op: write
##       Cur Op CV: [(0:1) (1:0) (2:1)]
##           Stack: Use -stacks to show stacks...
## =====================================================================
## 
\end{verbatim}
}
\hrule
\vspace{0.5cm}

That is, the flags both the fields \texttt{Test.y} and \texttt{Test/x}.
However, in any execution in \texttt{Test.x} both written and read by different threads,
the read always occurs before the write (with a read on \texttt{y} separating them).

\subsection{ThreadSanitizer}

We run the C program in Fig~\ref{simple_c} on ThreadSanitizer (shipped with LLVM).

\vspace{0.5cm}
\hrule
{\small
\begin{verbatim}
==================
WARNING: ThreadSanitizer: data race (pid=76224)
  Read of size 4 at 0x0001030ba074 by thread T2:
    #0 Thread2(void*) simple_race.cc:13 (a.out:x86_64+0x100000d1e)

  Previous write of size 4 at 0x0001030ba074 by thread T1:
    #0 Thread1(void*) simple_race.cc:8 (a.out:x86_64+0x100000cc9)

  Location is global 'y' at 0x0001030ba074 (a.out+0x000100001074)

  Thread T2 (tid=747232, running) created by main thread at:
    #0 pthread_create <null>:1600736 (libclang_rt.tsan_osx_dynamic.dylib:x86_64h+0x283ed)
    #1 main simple_race.cc:27 (a.out:x86_64+0x100000df3)

  Thread T1 (tid=747231, finished) created by main thread at:
    #0 pthread_create <null>:1600736 (libclang_rt.tsan_osx_dynamic.dylib:x86_64h+0x283ed)
    #1 main simple_race.cc:26 (a.out:x86_64+0x100000dd4)

SUMMARY: ThreadSanitizer: data race simple_race.cc:13 in Thread2(void*)
==================
==================
WARNING: ThreadSanitizer: data race (pid=76224)
  Write of size 4 at 0x0001030ba070 by thread T2:
    #0 Thread2(void*) simple_race.cc:14 (a.out:x86_64+0x100000d3a)

  Previous read of size 4 at 0x0001030ba070 by thread T1:
    #0 Thread1(void*) simple_race.cc:8 (a.out:x86_64+0x100000cae)

  Location is global 'x' at 0x0001030ba070 (a.out+0x000100001070)

  Thread T2 (tid=747232, running) created by main thread at:
    #0 pthread_create <null>:1600736 (libclang_rt.tsan_osx_dynamic.dylib:x86_64h+0x283ed)
    #1 main simple_race.cc:27 (a.out:x86_64+0x100000df3)

  Thread T1 (tid=747231, finished) created by main thread at:
    #0 pthread_create <null>:1600736 (libclang_rt.tsan_osx_dynamic.dylib:x86_64h+0x283ed)
    #1 main simple_race.cc:26 (a.out:x86_64+0x100000dd4)

SUMMARY: ThreadSanitizer: data race simple_race.cc:14 in Thread2(void*)
==================
ThreadSanitizer: reported 2 warnings
\end{verbatim}
}
\hrule
\vspace{0.5cm}

ThreadSanitizer also reports a race on both global locations \texttt{x} and \texttt{y}.

\subsection{Helgrind}

Helgrind is a data race detector integrated with Valgrind.
We analyzed the C program in Fig~\ref{simple_c}.
\vspace{0.5cm}
\hrule
% \vspace{0.5cm}
{\small
\begin{verbatim}
==2403== Helgrind, a thread error detector
==2403== Copyright (C) 2007-2015, and GNU GPL'd, by OpenWorks LLP et al.
==2403== Using Valgrind-3.12.0 and LibVEX; rerun with -h for copyright info
==2403== Command: ./a.out
==2403== 
==2403== ---Thread-Announcement------------------------------------------
==2403== 
==2403== Thread #3 was created
==2403==    at 0x596F30E: clone (in /usr/lib64/libc-2.17.so)
==2403==    by 0x4E41FD9: do_clone.constprop.4 (in /usr/lib64/libpthread-2.17.so)
==2403==    by 0x4E434C8: pthread_create@@GLIBC_2.2.5 (in /usr/lib64/libpthread-2.17.so)
==2403==    by 0x4C3064A: pthread_create_WRK (hg_intercepts.c:427)
==2403==    by 0x4C31728: pthread_create@* (hg_intercepts.c:460)
==2403==    by 0x400718: main (simple_race.cc:27)
==2403== 
==2403== ---Thread-Announcement------------------------------------------
==2403== 
==2403== Thread #2 was created
==2403==    at 0x596F30E: clone (in /usr/lib64/libc-2.17.so)
==2403==    by 0x4E41FD9: do_clone.constprop.4 (in /usr/lib64/libpthread-2.17.so)
==2403==    by 0x4E434C8: pthread_create@@GLIBC_2.2.5 (in /usr/lib64/libpthread-2.17.so)
==2403==    by 0x4C3064A: pthread_create_WRK (hg_intercepts.c:427)
==2403==    by 0x4C31728: pthread_create@* (hg_intercepts.c:460)
==2403==    by 0x4006F9: main (simple_race.cc:26)
==2403== 
==2403== ----------------------------------------------------------------
==2403== 
==2403== Possible data race during read of size 4 at 0x601044 by thread #3
==2403== Locks held: none
==2403==    at 0x4006A3: Thread2(void*) (simple_race.cc:13)
==2403==    by 0x4C3083E: mythread_wrapper (hg_intercepts.c:389)
==2403==    by 0x4E42E24: start_thread (in /usr/lib64/libpthread-2.17.so)
==2403==    by 0x596F34C: clone (in /usr/lib64/libc-2.17.so)
==2403== 
==2403== This conflicts with a previous write of size 4 by thread #2
==2403== Locks held: none
==2403==    at 0x40068E: Thread1(void*) (simple_race.cc:8)
==2403==    by 0x4C3083E: mythread_wrapper (hg_intercepts.c:389)
==2403==    by 0x4E42E24: start_thread (in /usr/lib64/libpthread-2.17.so)
==2403==    by 0x596F34C: clone (in /usr/lib64/libc-2.17.so)
==2403==  Address 0x601044 is 0 bytes inside data symbol "y"
==2403== 
==2403== ----------------------------------------------------------------
==2403== 
==2403== Possible data race during write of size 4 at 0x601040 by thread #3
==2403== Locks held: none
==2403==    at 0x4006AE: Thread2(void*) (simple_race.cc:14)
==2403==    by 0x4C3083E: mythread_wrapper (hg_intercepts.c:389)
==2403==    by 0x4E42E24: start_thread (in /usr/lib64/libpthread-2.17.so)
==2403==    by 0x596F34C: clone (in /usr/lib64/libc-2.17.so)
==2403== 
==2403== This conflicts with a previous read of size 4 by thread #2
==2403== Locks held: none
==2403==    at 0x400685: Thread1(void*) (simple_race.cc:8)
==2403==    by 0x4C3083E: mythread_wrapper (hg_intercepts.c:389)
==2403==    by 0x4E42E24: start_thread (in /usr/lib64/libpthread-2.17.so)
==2403==    by 0x596F34C: clone (in /usr/lib64/libc-2.17.so)
==2403==  Address 0x601040 is 0 bytes inside data symbol "x"
==2403== 
==2403== 
==2403== For counts of detected and suppressed errors, rerun with: -v
==2403== Use --history-level=approx or =none to gain increased speed, at
==2403== the cost of reduced accuracy of conflicting-access information
==2403== ERROR SUMMARY: 2 errors from 2 contexts (suppressed: 0 from 0)
\end{verbatim}
}

\hrule
\vspace{0.5cm}

Helgrind also incorrectly reports a race on both \texttt{x} and \texttt{y}.

\subsection{DRD}

Valgrind also provides another race detector DRD.
We analyze program in Fig~\ref{simple_c} in DRD:

% \hrule
\vspace{0.5cm}
\hrule
{\small
\begin{verbatim}
==2624== drd, a thread error detector
==2624== Copyright (C) 2006-2015, and GNU GPL'd, by Bart Van Assche.
==2624== Using Valgrind-3.12.0 and LibVEX; rerun with -h for copyright info
==2624== Command: ./a.out
==2624== 
==2624== Thread 3:
==2624== Conflicting load by thread 3 at 0x00601044 size 4
==2624==    at 0x4006A3: Thread2(void*) (simple_race.cc:13)
==2624==    by 0x4C30193: vgDrd_thread_wrapper (drd_pthread_intercepts.c:444)
==2624==    by 0x4E4FE24: start_thread (in /usr/lib64/libpthread-2.17.so)
==2624==    by 0x597C34C: clone (in /usr/lib64/libc-2.17.so)
==2624== Allocation context: BSS section of /home/umathur3/a.out
==2624== Other segment start (thread 2)
==2624==    (thread finished, call stack no longer available)
==2624== Other segment end (thread 2)
==2624==    (thread finished, call stack no longer available)
==2624== 
==2624== Conflicting store by thread 3 at 0x00601040 size 4
==2624==    at 0x4006AE: Thread2(void*) (simple_race.cc:14)
==2624==    by 0x4C30193: vgDrd_thread_wrapper (drd_pthread_intercepts.c:444)
==2624==    by 0x4E4FE24: start_thread (in /usr/lib64/libpthread-2.17.so)
==2624==    by 0x597C34C: clone (in /usr/lib64/libc-2.17.so)
==2624== Allocation context: BSS section of /home/umathur3/a.out
==2624== Other segment start (thread 2)
==2624==    (thread finished, call stack no longer available)
==2624== Other segment end (thread 2)
==2624==    (thread finished, call stack no longer available)
==2624== 
==2624== 
==2624== For counts of detected and suppressed errors, rerun with: -v
==2624== ERROR SUMMARY: 2 errors from 2 contexts (suppressed: 18 from 12)
\end{verbatim}
}

\hrule
\vspace{0.5cm}
Again, DRD reports two races, one of which is false.

\end{document}